\documentclass[10pt,onecolumn,a4paper]{IEEEtran}
\usepackage{soul}
\usepackage[mathscr]{eucal}
\usepackage[cmex10]{amsmath}
\usepackage{epsfig,epsf,psfrag}
\usepackage{amssymb,amsmath,amsthm,amsfonts,latexsym}
\usepackage{amsmath,graphicx,bm,xcolor,url,overpic}
\usepackage[caption=false]{subfig} 
\usepackage{fixltx2e}%ordering of single and double column floats
\usepackage{array}%array and tabular environments
\usepackage{verbatim}
\usepackage{bm}
\usepackage{algorithmic}
\usepackage{algorithm}
\usepackage{verbatim}
\usepackage{textcomp}
\usepackage{mathrsfs}
\usepackage{epstopdf}

\newcommand{\openone}{\leavevmode\hbox{\small1\normalsize\kern-.33em1}}

\catcode`~=11 \def\UrlSpecials{\do\~{\kern -.15em\lower .7ex\hbox{~}\kern .04em}} \catcode`~=13 

\allowdisplaybreaks[4]

\newcommand{\nn}{\nonumber}

\newcommand{\calA}{\mathcal{A}}
\newcommand{\calB}{\mathcal{B}}
\newcommand{\calC}{\mathcal{C}}

\newcommand{\calE}{\mathcal{E}}
\newcommand{\calF}{\mathcal{F}}
\newcommand{\calG}{\mathcal{G}}

\newcommand{\calI}{\mathcal{I}}
\newcommand{\calJ}{\mathcal{J}}

\newcommand{\calN}{\mathcal{N}}

\newcommand{\calP}{\mathcal{P}}

\newcommand{\calS}{\mathcal{S}}
\newcommand{\calT}{\mathcal{T}}
\newcommand{\calU}{\mathcal{U}}
\newcommand{\calV}{\mathcal{V}}

\newcommand{\calX}{\mathcal{X}}
\newcommand{\calY}{\mathcal{Y}}
\newcommand{\calZ}{\mathcal{Z}}

\newcommand{\bn}{\mathbf{n}}

\newcommand{\bs}{\mathbf{s}}
\newcommand{\bS}{\mathbf{S}}

\newcommand{\bv}{\mathbf{v}}
\newcommand{\bV}{\mathbf{V}}
\newcommand{\bw}{\mathbf{w}}

\newcommand{\bx}{\mathbf{x}}
\newcommand{\bX}{\mathbf{X}}

\newcommand{\rms}{\mathrm{s}}

\newcommand{\bbC}{\mathbb{C}}

\newcommand{\bbH}{\mathbb{H}}

\newcommand{\bbN}{\mathbb{N}}

\newcommand{\bbR}{\mathbb{R}}

\DeclareMathAlphabet{\mathbsf}{OT1}{cmss}{bx}{n}
\DeclareMathAlphabet{\mathssf}{OT1}{cmss}{m}{sl}% slanted sans serif

\DeclareSymbolFont{bsfletters}{OT1}{cmss}{bx}{n}  
\DeclareSymbolFont{ssfletters}{OT1}{cmss}{m}{n}
\DeclareMathSymbol{\bsfGamma}{0}{bsfletters}{'000}
\DeclareMathSymbol{\ssfGamma}{0}{ssfletters}{'000}
\DeclareMathSymbol{\bsfDelta}{0}{bsfletters}{'001}
\DeclareMathSymbol{\ssfDelta}{0}{ssfletters}{'001}
\DeclareMathSymbol{\bsfTheta}{0}{bsfletters}{'002}
\DeclareMathSymbol{\ssfTheta}{0}{ssfletters}{'002}
\DeclareMathSymbol{\bsfLambda}{0}{bsfletters}{'003}
\DeclareMathSymbol{\ssfLambda}{0}{ssfletters}{'003}
\DeclareMathSymbol{\bsfXi}{0}{bsfletters}{'004}
\DeclareMathSymbol{\ssfXi}{0}{ssfletters}{'004}
\DeclareMathSymbol{\bsfPi}{0}{bsfletters}{'005}
\DeclareMathSymbol{\ssfPi}{0}{ssfletters}{'005}
\DeclareMathSymbol{\bsfSigma}{0}{bsfletters}{'006}
\DeclareMathSymbol{\ssfSigma}{0}{ssfletters}{'006}
\DeclareMathSymbol{\bsfUpsilon}{0}{bsfletters}{'007}
\DeclareMathSymbol{\ssfUpsilon}{0}{ssfletters}{'007}
\DeclareMathSymbol{\bsfPhi}{0}{bsfletters}{'010}
\DeclareMathSymbol{\ssfPhi}{0}{ssfletters}{'010}
\DeclareMathSymbol{\bsfPsi}{0}{bsfletters}{'011}
\DeclareMathSymbol{\ssfPsi}{0}{ssfletters}{'011}
\DeclareMathSymbol{\bsfOmega}{0}{bsfletters}{'012}
\DeclareMathSymbol{\ssfOmega}{0}{ssfletters}{'012}

\newcommand{\hatbs}{\hat{\bs}}
\newcommand{\hatbS}{\hat{\bS}}

\newcommand{\hatbv}{\hat{\bv}}
\newcommand{\hatbV}{\hat{\bV}}

\newtheorem{theorem}{Theorem}

\newtheorem{definition}{Definition}

\graphicspath{{./figures/}}
\usepackage{graphicx,cite}
\usepackage{epstopdf}
\usepackage{enumerate}
\usepackage{color}
\usepackage{bbm}
\usepackage{algorithm}
\usepackage{algorithmic}
\usepackage[T1]{fontenc}
\usepackage{aecompl}
\usepackage{booktabs}
\usepackage{tabularx}
\usepackage[ colorlinks = true,
linkcolor = blue,
urlcolor  = blue,
citecolor = red,
anchorcolor = green,]{hyperref}

\newcommand{\bbo}{\mathbbm{1}}
\def\BibTeX{{\rm B\kern-.05em{\sc i\kern-.025em b}\kern-.08em T\kern-.1667em\lower.7ex\hbox{E}\kern-.125emX}}

\IEEEoverridecommandlockouts
\definecolor{Dyellow}{RGB}{254,152,0}
\definecolor{Dgreen}{RGB}{0,176,80}

\newcommand{\blue}[1]{\textcolor{blue}{#1}}%Replace original words.
\makeatletter

\newcommand{\Rmnum}[1]{\expandafter\@slowromancap\romannumeral #1@}
\makeatother

\def\BibTeX{{\rm B\kern-.05em{\sc i\kern-.025em b}\kern-.08em
T\kern-.1667em\lower.7ex\hbox{E}\kern-.125emX}}

\begin{document}
\title{Resolution Limits of Non-Adaptive 20 Questions Estimation for Tracking Multiple Moving Targets\\
}

\author{Chunsong Sun, Lin Zhou, Jingjing Wang, Weijie Yuan, Chunxiao Jiang and Alfred Hero\\

\thanks{Parts of this paper were submitted to IEEE ITW 2025.}

\thanks{C. Sun, L. Zhou and J. Wang are with the School of Cyber Science and Technology, Beihang University, Beijing, China, 100191, (Emails: \{sunchunsong, lzhou, drwangjj\}@buaa.edu.cn). W. Yuan is with School of Automation and Intelligent Manufacturing, Southern University of Science and Technology, Shenzhen, China, 518055, (Email: yuanwj@sustech.edu.cn). C. Jiang is with Beijing National Research Center for Information Science and Technology (BNRist), Tsinghua University, Beijing, China, 100084 (Email: jchx@tsinghua.edu.cn). Alfred Hero is with the Department of Electrical Engineering and Computer Science, University of Michigan, Ann Arbor, MI, USA, 48109-2122 (Email: hero@eecs.umich.edu.).}

% \thanks{L. Zhou was supported by the National Natural Science Foundation of China under Grant 62201022, the Beijing Municipal Natural Science Foundation under Grant 4232007 and the start up grant of Beihang University. J. Wang was supported by the National Natural Science Foundation of China under Grant 62222101 and Beijing Municipal Natural Science Foundation under Grants L232043 and L222039. C. Jiang was supported by the National Natural Science Foundation of China under Grants 62325108 and 62341131. Alfred O. Hero III was supported by the United States National Science Foundation under Grant NSF CCF 2246213.}
}

\maketitle

\begin{abstract}
Motivated by the practical application of beam tracking of multiple devices in Multiple Input Multiple Output (MIMO) communication, we study the problem of non-adaptive twenty questions estimation for locating and tracking multiple moving targets under a query-dependent noisy channel. Specifically, we derive a non-asymptotic bound and a second-order asymptotic bound on resolution for optimal query procedures and provide numerical examples to illustrate our results. In particular, we demonstrate that the bound is achieved by a state estimator that thresholds the mutual information density over possible target locations. This single threshold decoding rule has reduced the computational complexity compared to the multiple threshold scheme proposed for locating multiple stationary targets (Zhou, Bai and Hero, TIT 2022). We discuss two special cases of our setting: the case with unknown initial location and known velocity, and the case with known initial location and unknown velocity. Both cases share the same theoretical benchmark {that applies to} stationary multiple target search in Zhou, Bai and Hero (TIT 2022) while the known initial location case is close to the theoretical benchmark for stationary target search when the maximal speed is inversely proportional to the number of queries. We also generalize our results to account for a piecewise constant velocity model introduced in Zhou and Hero (TIT 2023), where targets change velocity periodically. Finally, we illustrate our proposed algorithm for the application of beam tracking of multiple mobile transmitters in a 5G wireless network.
\end{abstract}


\section{Introduction} \label{sec_intro}

Twenty questions estimation was pioneered by R\'enyi~\cite{renyi1961problem} and Ulam~\cite{ulam1991adventures}, who mathematically modeled a two-player question and answer game. Over the following decades, these results have been generalized to cases where the query procedures are adaptive~\cite{chiu2019noisy} or non-adaptive~\cite{jedynak2012twenty}, and the noise that corrupts the oracle’s responses is query-dependent~\cite{kaspi2017searching,lalitha2018improved} or query-independent~\cite{jedynak2012twenty,chung2017unequal}. However, most studies on twenty questions estimation consider a single stationary target~\cite{chung2017unequal,chiu2016sequential,chiu2019noisy}. Kaspi \emph{et al.}~\cite[Section V]{kaspi2017searching} studied the case of one-dimensional target moving over the unit circle and derived first-order asymptotic bounds on the performance of an optimal non-adaptive search strategy with an infinite number of queries over a query-dependent binary symmetric noisy channel. Recently, Zhou and Hero~\cite{zhou2023resolution} refined~\cite[Section V]{kaspi2017searching} by deriving second-order asymptotic bounds when searching for a moving target over a finite dimensional unit cube under any query-dependent discrete memoryless channel (DMC) and additive white Gaussian noise (AWGN) channel. In~\cite{zhou2022resolution}, the framework of~\cite{kaspi2017searching} was extended to multiple stationary targets.

In this paper, we extend the results of~\cite{zhou2022resolution,zhou2023resolution} to the case of tracking multiple moving targets. This is a non-trivial extension due to the higher dimensional search space and the significant increase in computational complexity. Specifically, we introduce an extension of the bounds and a non-adaptive tracking algorithm with lower complexity that uses a single detection threshold on the mutual information density as contrasted to the multiple threshold scheme proposed in~\cite{zhou2022resolution}.

As in~\cite{kaspi2017searching,zhou2022resolution,zhou2023resolution}, we adopt the non-adaptive twenty questions querying framework. Specifically, extending the trajectory analysis in~\cite{zhou2023resolution} for a single moving target, we obtain a non-adaptive tracking procedure for multiple moving targets and derive non-asymptotic and second-order asymptotic bounds on achievable resolution. Our proposed query and target estimation procedure are inspired by the coding scheme for variable length sparse feedback multiple access protocol introduced in~\cite{yavas2023variable}, where the single mutual information density threshold rule is used to detect target state. Our second-order analysis reveals that searching for multiple moving targets is fundamentally twice as difficult as searching for multiple stationary targets~\cite{zhou2023resolution}, and reveals how the number of targets impacts the theoretical benchmarks as compared to multiple stationary target search~\cite{zhou2022resolution}. Furthermore, we discuss two special cases: the case with known initial location and unknown velocity, and the case with unknown initial location and known velocity. Both cases share the same theoretical benchmark as the case of stationary target search~\cite{zhou2022resolution}, while the second case establishes a theoretical benchmark of stationary target search when the maximal speed satisfies certain conditions. We also generalize our results to account for the piecewise constant velocity model by Zhou and Hero~\cite{zhou2023resolution}, where each target changes velocity periodically. Finally, we demonstrate how our twenty questions solution for multiple target tracking can be used in the practical application of beam tracking for 5G communications, providing a use-case motivating our extension of~\cite{zhou2021resolution,zhou2022resolution,zhou2023resolution}.

\section{Problem Formulation} \label{sec_pf}

\subsection*{Notation} \label{notation}

Random variables and their realizations are denoted by upper case variables (e.g., $X$) and lower case variables (e.g., $x$), respectively. All sets are denoted in calligraphic font (e.g., $\calX$). We use $\bbC, \bbR, \bbR_+$ and $\bbN$ to denote the set of complex numbers, real numbers, positive real numbers, and positive integers, respectively. Given any positive integer $n\in\bbN$, let $X^n:=(X_1,\ldots,X_n)$ be a random vector of length $n$. We use $\Phi^{-1}(\cdot)$ to denote the inverse of the cumulative distribution function (cdf) of the standard Gaussian distribution. Given any two positive integers $(m, n)\in\bbN^2$, we use $[m : n]$ to denote the set $\{m, m + 1,\ldots,n\}$ and use $[m]$ to denote $[1 : m]$. The collection of non-empty strict subset of set $\calJ$ is denoted by $\Lambda(\calJ):=\{\calJ^\prime:\calJ^\prime \subseteq \calJ, 0 < |\calJ^\prime| < |\calJ|\}$. Given any $m\in\bbN$, for any length-$m$ vector $a = (a_1, \ldots , a_m)$, the infinity norm is defined as $\|a\|_\infty:=\max_{i\in[m]}|a_i|$. The set of all probability distributions on a finite set $\calX$ is denoted as $\calP(\calX)$ and the set of all conditional probability distributions from $\calX$ to $\calY$ is denoted as $\calP(\calY|\calX)$. Furthermore, we use $\calF(\calS\times\calV)$ to denote the set of all probability density functions (pdf) defined on the set $\calS\times\calV$. Let Bern$(p)$ denote the Bernoulli distribution with parameter $p\in(0,1)$. All logarithms are base $e$. We use $\bbo(\cdot)$ to denote the indicator function. The symbols used in this paper are summarized in Table \ref{table_sym}.

\begin{table}[tb]
\caption{Symbol Table}
\label{table_sym}
\centering
\begin{tabular}{|m{1cm}<{\centering}|m{6.5cm}|m{1cm}<{\centering}|m{6.5cm}|}
\hline
Symbol & Meaning & Symbol & Meaning\\
\hline
$k$ & Number of targets & $d$ & Dimension of a target\\
\hline
$\bS$ & Initial location of a target & $\bV$ & Velocity of a target\\
\hline
$l$ & Real-time location of a target & $f$ & Probability density function\\
\hline
$\calA$ & Query region & $n$ & Number of queries\\
\hline
$h$ & Lipschitz continuous function & $\delta$ & Tolerable resolution\\
\hline
$\varepsilon$ & Excess-resolution probability & $v_+$ & Maximum speed of moving targets\\
\hline
$M$ & Number of sub-intervals in each dimension & $\Gamma$ & Transform function\\
\hline
$\mathrm{q}$ & Quantization function & $\bw$ & Quantized trajectory of a target\\
\hline
$\calB_{n,M}$ & Set of all possible quantized trajectories & $\calB^{(k)}_{n,M}$ & Set of all possible $k$ distinct quantized trajectories\\
\hline
$\calU_{n,M}$ & Set of given $k$ targets' quantized trajectories & $\imath^{h(p)}$ & Mutual information density\\
\hline
$p$ & Parameter of Bernoulli distribution for query generation & $\sigma$ & Standard deviation of Gaussian noise\\
\hline
$\delta^*$ & Minimal achievable resolution & $\gamma$ & Decoding threshold\\
\hline
$C(p,t)$ & Mean of mutual information density & $V(p,t)$ & Variance of mutual information density\\
\hline
$C(k)$ & Capacity of the query-dependent channel & $V(k,\varepsilon)$ & Dispersion of the query-dependent channel\\
\hline
\end{tabular}
\end{table}

\subsection{System Model} \label{pf System Model}

Fix the domain dimension $d \in \bbN$ and the maximum speed $v_+ \in \bbR_+$ of $k$ targets moving in the domain, which we assume is the $d$-dimensional cube. The goal of non-adaptive search for multiple moving targets is to estimate the trajectory of the moving targets given unknown initial locations $\bs^k=(\bs_1,\ldots,\bs_k)\in([0,1]^d)^k$ and velocities $\bv^k=(\bv_1,\ldots,\bv_k)\in([-v_+,v_+]^d)^k$.

Consistent with~\cite[Section II-A]{zhou2023resolution}, we make the following assumptions. Given any $\mathbf{s}\in[0,1]^d$ and $\mathbf{v}\in[-v_+,v_+]^d$, if a target moves with a constant velocity $\mathbf{v}$ from an initial location $\mathbf{s}$, at any time point $i\in\bbN$ and dimension $j\in[d]$, its location $l(s_j,v_j,i)$ at time $i$ satisfies
\begin{align}
l\big(s_j,v_j,i\big)
&:=\bigg\{
\begin{aligned}
&\mathrm{mod}(s_j+iv_j,2)&\mathrm{if}~\mathrm{mod}(s_j+iv_j,2) \leq 1,\\
&2-\mathrm{mod}(s_j+iv_j,2)&\mathrm{if}~\mathrm{mod}(s_j+iv_j,2) > 1,
\end{aligned}
\label{loc}
\end{align}
where $\mathrm{mod}(a,b)$ refers to the modulo operator. The real-time location of each target is a $d$-dimensional vector $l(\mathbf{s},\mathbf{v},i)=(l(s_1,v_1,i),\ldots,l(s_d,v_d,i))$.

Let $\calS=[0,1]^d$ and $\calV=[-v_+,v_+]^d$. Furthermore, let $(\bS^k,\bV^k)=(\bS_1,\ldots,\bS_k,\bV_1,\ldots,\bV_k)$ be $k$ pairs of location and velocity random variables generated from an arbitrary and unknown pdf $f_{\bS^k\bV^k}\in\calF(\calS^k\times\calV^k)$. Assume that there is an oracle that knows the instantaneous location of all targets at any time. Fix any $i\in\bbN$. At time point $i$, the questioner poses a query asking if there is any target in a Lebesgue measurable query set $\calA_i \subseteq [0,1]^d$. For each $t\in[k]$, the indicator function $X_i(t):=\bbo(l(\mathbf{S}_t,\mathbf{V}_t,i)\in\calA_i)$ denotes the binary yes/no answer. After receiving the query, the oracle generates a binary noiseless response $Z_i$ from the binary OR of $(X_i(1),\ldots,X_i(k))$. Subsequently, the noiseless response $Z_i$ is transmitted over a binary input point-to-point query-dependent channel with transition probability $P_{Y_i|Z_i} \in \calP(\calY_i|\{0,1\})$, yielding a noisy response $Y_i$. At a predefined maximum time point $n\in \bbN$, given noisy responses $Y^n=(Y_1,\ldots,Y_n)$, the questioner forms the final estimate of the targets' trajectory via a decoding function.

\subsection{Query-Dependent Channel} \label{pf Query-Dependent Channel}

We briefly recall the definition of a query-dependent channel in~\cite{kaspi2017searching,chiu2016sequential,zhou2021resolution}, which is also known as a channel with state~\cite[Chapter 7]{el2011network}. Assume that the query-dependent channel depends on the query only through its size. Let $h:~[0,1]\to\bbR_+$ be a bounded Lipschitz continuous function with parameter $L$, i.e., $|h(p_1) - h(p_2)| \leq L|p_1 - p_2|$ for any $(p_1, p_2) \in [0, 1]^2$ and $\max_{p \in [0,1]} h(p) < \infty$. Given a query $\calA$, the channel from the oracle to the questioner is a query-dependent channel, where the channel depends on the query through its size $|\calA| = \int_{m \in \calA} \mathrm{dm}$ via the function $h(\cdot)$. Specifically, the channel can be represented by a state-dependent channel $P_{Y|Z}^{h(|\calA|)}$. When $h(\cdot)$ maps to a constant, the above channel model specializes to a query-independent channel~\cite{chung2017unequal,jedynak2012twenty} that does not depend on $\calA$.

Let $q\in[0, 1]$. For any $\xi\in(0, \min\{q, 1-q\})$, analogous to~\cite[Section II-B]{zhou2021resolution}, we assume that the query-dependent channel is continuous in the sense that there exists a constant $c(q)$ depending on $q$ only such that
\begin{align}
&\max \Bigg\{ {{\Bigg\| \Bigg\{\log \frac{{P_{Y|Z}^q}(z,y)}{P_{Y|Z}^{q + \xi}(z,y)}\Bigg\}_{(z,y)\in\calZ\times\calY} \Bigg\|}_\infty},~{{\Bigg\| \Bigg\{\log \frac{{P_{Y|Z}^q}(z,y)}{P_{Y|Z}^{q - \xi}(z,y)}\Bigg\}_{(z,y)\in\calZ\times\calY} \Bigg\|}_\infty} \Bigg\} \leq c( q )\xi. \label{dmc}
\end{align}

Two common examples of the query-dependent channel are as follows. The first type of channel is discrete with query-dependent Bernoulli noise.

\begin{definition} \label{def1}
Given any $\calA \subseteq [0,1]$, a channel $P_{Y|Z}^{h(|\calA|)}$ is said to be a query-dependent binary symmetric channel (BSC) if $\calZ = \calY = \{0, 1\}$ and $\forall~ (z,y)\in\{0, 1\}^2$,
\begin{align}
P_{Y|Z}^{h(|\calA|)}=\big(h(|\calA|)\big)^{\bbo(y \neq z)}\big(1-h(|\calA|)\big)^{\bbo(y = z)}.\label{bsc}
\end{align} 
\end{definition}

The query size function $h(\cdot)$ determines the crossover probability of the above BSC. Smaller size query sets correspond to lower crossover probability when $h(\cdot)$ is an increasing function. The second channel is continuous with query-dependent AWGN, which will be used in our numerical illustrations below.

\begin{definition} \label{def2}
Given any $\calA \subseteq [0, 1]$, a channel $P_{Y|Z}^{h(|\calA|)}$ is said to be a query-dependent AWGN channel with parameter $\sigma\in\bbR_+$ if $\calZ = \{0, 1\}, \calY = \bbR$ and $\forall~ (z,y)\in\{0, 1\}\times \bbR$,
\begin{align}
P_{Y|Z}^{h(|\calA|)} = \frac{1}{\sqrt{2\pi(h(|\calA|)\sigma)^2}}\exp\Big(-\frac{(y-z)^2}{2(h(|\calA|)\sigma)^2}\Big).
\end{align}
\end{definition}

Here the query size function $h(\cdot)$ modulates the noise variance of the above AWGN channel. Smaller query sets $\calA$ correspond to higher signal to noise ratio when $h(\cdot)$ is an increasing function. The query-dependent AWGN channel violates the constraint in \eqref{dmc} since $(y-z)^2$ can be unbounded. However, the constraint in Eq. \eqref{dmc} can be made to hold with high probability by restricting $(z,y)$ to satisfy $(y-z)^2<D$, where $D\in\bbR_+$ is arbitrarily large but finite similarly to~\cite[Section II-B]{zhou2023resolution}.

\subsection{Definition of Query Procedure} \label{pf Non-Adaptive Query Procedure}

A non-adaptive query procedure for multiple moving target search is defined as follows.

\begin{definition} \label{def3}
Given any three positive integers $(n,k,d)\in\bbN^3$ and two positive real numbers $(\delta,\varepsilon)\in\bbR_+\times(0,1)$, an $(n,k,d,\delta,\varepsilon)$-non-adaptive query procedure consists of
\begin{itemize}
\item[$\bullet$] $n$ query sets $(\calA_1,\dots,\calA_n) \subseteq ([0,1]^d)^n$,
\item[$\bullet$] a decoding function $g:\calY^n \rightarrow (\calS\times \calV)^k$,
\end{itemize}
such that the excess-resolution probability satisfies
\begin{align}
\mathrm{P_e}(n,k,d,\delta):=\sup_{f_{\bS^k\bV^k}\in\calF(\calS^k\times\calV^k)} \mathrm{Pr}\bigg\{ \exists~t\in[k]:\min_{(\hatbs,\hatbv)\in(\hatbs,\hatbv)^k} \max_{i\in[0:n]} \| l(\hatbs,\hatbv,i)-l(\bS_t,\bV_t,i) \|_{\infty} > \delta \bigg\} \leq \varepsilon , \label{pe}
\end{align}
where $(\hatbs,\hatbv)^k = \{(\hatbs_1,\hatbv_1),\ldots,(\hatbs_k,\hatbv_k)\}$ is the set of estimated initial location and velocity pairs of the targets.
\end{definition}

Given any $(n,k,d,\delta,\varepsilon)$ query procedure, with probability of at least $1-\varepsilon$, one can estimate the trajectories of all targets accurately within all $n$ time slots  with resolution $\delta$. Analogous to~\cite{zhou2022resolution,zhou2023resolution}, the theoretical benchmark for non-adaptive querying for multiple moving targets is defined as the following minimal achievable resolution for any $(n,k,d,\varepsilon)\in\bbN^3\times(0,1)$:
\begin{align}
\delta^*(n,k,d,\varepsilon):=\inf\big\{\delta\in\bbR_+:~\exists~\mathrm{an}~(n,k,d,\delta,\varepsilon)\textrm{-non-adaptive~query~procedure} \big\}.
\end{align}
When there is only one target, i.e., $k = 1$, the above definitions specialize to non-adaptive search for a single moving target~\cite{zhou2023resolution} with time slot $B=1$. When the maximum speed $v_+=0$, the above definitions specialize to multiple stationary target search~\cite{zhou2022resolution} with distinct initial locations.

\section{Main Results} \label{sec_mr}

\subsection{Preliminary Definitions} \label{mr Preliminary Definitions}

This subsection presents definitions needed to present our query procedure and its mathematical analysis. The first several definitions are analogous to the definitions of the trajectory part in~\cite[Section III]{zhou2023resolution}. Let $n$ denote the number of queries we can pose and let $M$ denote the number of sub-intervals at each dimension, which we use to quantize the target locations. Given any positive integers $(d,n,M)\in\bbN^3$, define a function $\Gamma:[nM]^d \to [(nM)^d]$ such that for any $(j_1,\ldots,j_d)\in[nM]^d$,
\begin{align}
\Gamma(j_1,\ldots,j_d):=1+\sum_{d^{\prime}\in[d]}(j_{d^{\prime}}-1)(nM)^{d-d^{\prime}}. \label{gamma-func}
\end{align}
Furthermore, given any $s\in[0,1]$ and $n\in\bbN$, define the quantization function
\begin{align}
\mathrm{q}(s,n):=\lceil snM \rceil.
\end{align}
The function $\mathrm{q}(s,n)$ is used to discretize the location into a certain interval so that non-adaptive grid search can be applied. Suppose that there is a target with initial location $\mathbf{s}=(s_1,\ldots,s_d)\in[0,1]^d$ and velocity $\mathbf{v}=(v_1,\ldots,v_d)\in[-v_+,v_+]^d$. For each $(i,j)\in[n]\times[d]$, the quantized location of each target at time $i$ along the $j$-th dimension satisfies
\begin{align}
w(s_j,v_j,i):=\mathrm{q}\big(l(s_j,v_j,i),n\big) , \label{quantized-loc}
\end{align}
and the quantized location of each target is a $d$-dimensional vector $w(\mathbf{s},\mathbf{v},i)=(w(s_1,v_1,i),\ldots,w(s_d,v_d,i))$.

Efficient search for a moving target requires estimation of its trajectory. Fix any $n\in\bbN$. Define the set of all possible quantized trajectories of a moving target within $n$ time points as
\begin{align}
\calB_{n,M}
&:=\big\{\bw^n\in([nM]^d)^n: \bw^n=\bw(\mathbf{s},\mathbf{v},[n])~\mathrm{for~some}~(\mathbf{s},\mathbf{v})\in[0,1]^d \times [-v_+,v_+]^d\big\} , \label{Bnm}
\end{align}
where $\bw(\mathbf{s},\mathbf{v},[n])=(w(\mathbf{s},\mathbf{v},1),\dots,w(\mathbf{s},\mathbf{v},n))\in([nM]^d)^n$ is a trajectory with initial location $\mathbf{s}$ and velocity $\mathbf{v}$, i.e., the collection of quantized locations at discrete time points $[n]$ of a target. Analogous to~\cite[Eq. 11]{zhou2023resolution}, the size of $\calB_{n,M}$ is upper bounded by
\begin{align}
|\calB_{n,M}| \leq \big((2nv_++3)n^4M^2\big)^d. \label{bnm}
\end{align}

In subsequent analyses, we assume that $k$ targets are pairwise distinct, i.e. there are no two targets that have the same quantized trajectory. Our analysis can also be generalized to account for the case where the quantized trajectories of two targets overlap, see Fig. \ref{fig1}. In this case, the excess-resolution event in Section \ref{Achievable Non-Asymptotic Bound DMC} needs to modified to account for the fact that the number of estimated trajectories is smaller than the number of distinct trajectories $k$.

\begin{figure}[tb]
\centering
\includegraphics[width=.5\columnwidth]{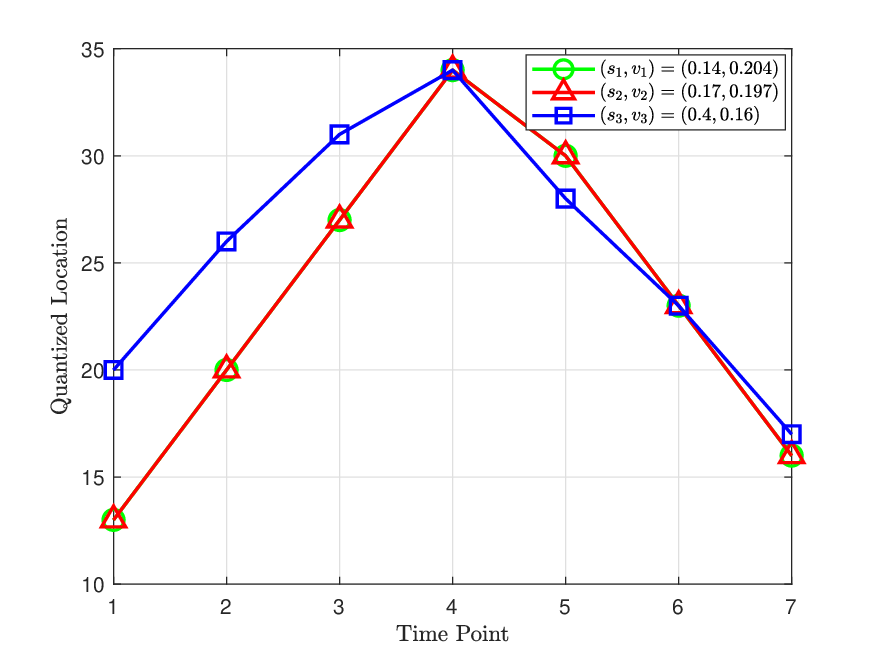}
\caption{Illustration of quantized trajectories for three targets where $2$ of the targets have trajectories (red triangles and green circles) that, after quantization, cannot be distinguished from each other. Here the number of queries is $n=7$ and the number of quantization cells per dimension is $M=5$.}
\label{fig1}
\end{figure}

Next, define the set of $k$ distinct trajectories sets as
\begin{align}
\calB_{n,M}^{(k)}:=\big\{\{\bw_1^n,\ldots,\bw_k^n\}:\bw_t^n\in\calB_{n,M}~\mathrm{for~any}~t\in[k]~\mathrm{and}~\bw_{t_1}^n \neq \bw_{t_2}^n~\mathrm{for~any}~(t_1,t_2)\in[k]^2 \big\}.
\end{align}
For any $k$ moving targets with initial location and velocity vectors $(\mathbf{s},\mathbf{v})^k\in([0,1]^d \times [-v_+,v_+]^d)^k$, define the set of targets' quantized trajectories as
\begin{align}
\calU_{n,M}\left((\mathbf{s},\mathbf{v})^k\right):=\big\{\mathbf{w}(\mathbf{s}_1,\mathbf{v}_1,[n]),\ldots,\mathbf{w}(\mathbf{s}_k,\mathbf{v}_k,[n])\big\}. \label{set-u}
\end{align}

The next definitions are analogous to the mutual information density in~\cite[Section III]{zhou2022resolution}. Given any sequence of random variables $X_{[k]} = (X_1,\ldots,X_k)$, we denote the binary OR of $X_{[k]}$ as the symbol $Z$. For any $(x_{[k]}, y)\in\{0, 1\}^k \times \calY$, define the joint distribution
\begin{align}
P_{X_{[k]} Y}^{h(p)}\big(x_{[k]}, y\big)
&:=\bigg(\prod_{t\in[k]} P_X(x_t)\bigg) P_{Y | Z}^{h(p)}(y | z) , \label{joint_d}
\end{align}
where $p\in(0,1)$, $P_X=\mathrm{Bern}(p)$ and $P_{Y | Z}^{h(p)}$ denotes the binary input query-dependent channel. Fix any $t\in[k]$. Let $P_{Y|X_{[t]}}$ be induced by $P_{X_{[k]} Y}^{h(p)}$. Given any $t\in[k]$ and any $p\in(0,1)$, define the mutual information density
\begin{align}
\imath^{h(p)}(x_{[t]} ; y):=\log \frac{P_{Y | X_{[t]}}^{h(p)}(y | x_{[t]})}{P_Y^{h(p)}(y)}. \label{mu_info_e}
\end{align} 
Given any query region $\calA \in [0,1]^d$, for any $(z,y)\in\{0,1\} \times \calY$, define the query-dependent mutual information density 
\begin{align}
\imath_{\calA}(z;y):=\log\frac{P_{Y|Z}^{h(|\calA|)}(y|z)}{P_{Y}^{h(|\calA|)}(y)} . \label{mu_info_a}
\end{align}

Given $(p,t)\in(0,1)\times[k]$, let $C(p,t)$ and $V(p,t)$ denote the mean and variance of the mutual information density $\imath^{h(p)}$, respectively, i.e.
\begin{align}
C(p,t)&:=\mathbb{E}\big[\imath^{h(p)}\big(X_{[t]};Y\big)\big] , \label{cempty}\\
V(p,t)&:=\mathrm{Var}\big[\imath^{h(p)}\big(X_{[t]};Y\big)\big] . \label{vempty}
\end{align}
Furthermore, the capacity of the binary-input multiple access channel (MAC) induced by \eqref{joint_d} is defined as
\begin{align}
C(k):=\max_{p\in(0,1)} C(p,k). \label{ck}
\end{align}
Let $\calP_k$ denote the set of capacity achieving optimizer $p^*$ in \eqref{ck}. Given any $\varepsilon \in (0,1)$, the dispersion for the channel is then defined as
\begin{align}
V(k,\varepsilon):=\Bigg\{
\begin{aligned}
&\min_{p^*\in\calP_k} V(p,k)&\mathrm{if}~\varepsilon \leq 0.5, \\
&\max_{p^*\in\calP_k} V(p,k)&\mathrm{if}~\varepsilon > 0.5.
\end{aligned}
\label{vk}
\end{align}

\subsection{Non-Asymptotic Bounds} \label{mr Non-Asymptotic Bounds}

\begin{algorithm}[tb]
\caption{Non-Adaptive Query Procedure for Multiple Moving Target Search}
\label{alg:1} 
\begin{algorithmic}[1]  
\REQUIRE
The number of queries $n\in\bbN$, dimension $d\in\bbN$ and three design parameters $(M,p,\gamma)\in\bbN \times (0,1) \times \bbR_+$
\ENSURE
A set of estimated trajectories $\{\hat{\bw}_1^n,\ldots,\hat{\bw}_k^n\}$ of $k$ targets with unknown initial locations and velocities
~\\
\hrule 
~\\
\STATE Divide the $d$-dimensional unit cube $\calS=[0,1]^d$ into $(nM)^d$ equally sized non-overlapping sub-cubes $(\calC_1,\ldots,\calC_{(nM)^d})$.
\STATE \textbf{\emph{Query generation and response collection:}}
\STATE Generate $n$ binary vectors $(x^n(1),\ldots,x^n((nM)^d))$ with length-$(nM)^d$, where each vector is generated i.i.d. from Bern$(p)$.
\FOR{$i\in[n]$}
\STATE Form a query region $\calA_i$ as
\begin{align}
\calA_i:=\bigcup_{j\in[(nM)^d]:~x_{i,j}=1} \calC_j \nn.
\end{align}
\STATE Pose the $i$-th query asking whether there are targets in the region $\calA_i$ and obtain a noisy response $y_i$ from the oracle.
\ENDFOR
\STATE Collect noisy responses $y^n=(y_1,\ldots,y_n)$.
\STATE \textbf{\emph{Decoding:}}
\IF{there exists a set of trajectories $\{\hat{\bw}_1^n,\ldots,\hat{\bw}_k^n\}\in\calB_{n,M}^{(k)}$ such that
\begin{align}
\imath^{h(p)}\big(x^n(\Gamma(\hat{\bw}_1^n)),\ldots,x^n(\Gamma(\hat{\bw}_k^n)); y^n\big) \geq \gamma, \label{single_thres}
\end{align}
}
\RETURN estimates $\{ (\hatbs_1,\hatbv_1),\ldots,(\hatbs_k,\hatbv_k) \}$ such that for any $t\in[k]$, $$w(\hatbs_t,\hatbv_t,[n])=\hat{\bw}_t^n.$$
\ELSE
\RETURN estimates $\{ (\hatbs_1,\hatbv_1),\ldots,(\hatbs_k,\hatbv_k) \}$ such that $(\hatbs_t,\hatbv_t)= (\mathbf{0.5}_{d \times 1},\mathbf{0}_{d \times 1})$ for each $t\in[k]$.
\ENDIF
\end{algorithmic}
\end{algorithm}

Recall that $h(\cdot)$ is a bounded Lipschitz continuous function with parameter $L$. Recall the definitions of the function $\Gamma(\cdot)$ in \eqref{gamma-func}, the trajectory set $\calB_{n,M}$ in \eqref{bnm} and the power set $\Lambda([k])$ in the Notation Section, respectively. Fix $p\in(0,1)$. Fix positive real numbers $(\gamma,\eta)$ and $\{\lambda_\calJ\}_{\calJ\in\Lambda([k])}$. Let $(x^n(1),\ldots,x^n((nM)^d))$ be a sequence of $n$ binary codewords with length-$(nM)^d$, where each codeword is generated i.i.d. from $P_X=\mathrm{Bern}(p)$. For ease of notation, given any trajectory $\bw^n\in\calB_{n,M}$, we use $x^n(\Gamma(\bw^n))$ to denote $(x_{1,\Gamma(w_1)},\ldots,x_{n,\Gamma(w_n)})$. For the convenience of showing our non-asymptotic achievability results, given any $(n,\gamma,\lambda_{\calJ})\in\bbN \times \bbR_+^2$ and $\calJ\in\Lambda([k])$ and random vectors $(X^n(1),\ldots,X^n(k),Y^n)$ that are generated i.i.d. from $P_{X_{[k]}Y}^{h(p)}$, define the following two events:
\begin{align}
\calG_1(n,\gamma)
&:=\big\{ \imath^{h(p)}\big(X^n(1),\ldots,X^n(k);Y^n\big) < \gamma \big\} , \label{calG1}\\
\calG_2(n,\calJ,\lambda_{\calJ})&:=\big\{ \imath^{h(p)}\big(X_{\calJ}^n;Y^n\big) > nC(p,|\calJ|)+n\lambda_\calJ \big\}\label{calG2},
\end{align}
where $X_\calJ^n=\{X_i^n\}_{i\in\calJ}$.

Furthermore, given any $(a,b,\gamma)\in\bbN^2 \times \bbR_+$ such that $a\geq b$, define the following two functions:
\begin{align}
\bbH_1(a,b,\gamma)&:=\binom{a}{b} \exp(-\gamma) , \label{calH1}\\
\bbH_2(a,b,\gamma)&:=a^b \exp(-\gamma),\label{calH2}
\end{align}
where $\binom{a}{b}:=\frac{a!}{b!(a-b)!}$ is the binomial coefficient.

Fix positive integers $(n,k,d,M)\in\bbN^4$. Recall the definitions of the function $c(\cdot)$ in \eqref{dmc} and the joint distribution $P_{X_{[k]} Y}^{h(p)}$ in \eqref{joint_d}. Recall that $L\in\bbR_+$ is the parameter for the Lipschitz continuous function. In our achievability analyses, we assume that the parameter $M$, which is used to quantize the target locations, is much greater than the number of targets $k$.

\begin{theorem} \label{th1}
For any discrete query-dependent channel satisfying the conditions in \eqref{dmc}, there exists an $(n,k,d,\frac{2}{M},\varepsilon)$-non-adaptive query procedure such that
\begin{align}
\varepsilon &\leq 4n\exp\big(-2n^dM^d\eta^2\big) + \exp\big(n \eta Lc(h(p))\big) \times \bigg( \underset{(P_{X_{[k]} Y}^{h(p)})^n}{\mathrm{Pr}}\big\{ \calG_1(n,\gamma) \big\} + \bbH_1\Big(\big((2nv_++3)n^4M^2\big)^d-k,k,\gamma\Big) \bigg. \nn\\*
&\qquad\bigg.+ \sum_{\calJ\in\Lambda([k])} \bigg( \underset{(P_{X_{[k]} Y}^{h(p)})^n}{\mathrm{Pr}}\big\{ \calG_2(n,\calJ,\lambda_{\calJ}) \big\} + \bbH_1\Big(\big((2nv_++3)n^4M^2\big)^d-k,k-|\calJ|,\gamma-nC(p,|\calJ|)-n\lambda_\calJ\Big) \bigg) \bigg). \label{th1-rlt} 
\end{align}
\end{theorem}

The proof of Theorem \ref{th1} is provided in Section \ref{Achievable Non-Asymptotic Bound DMC}. The non-asymptotic bound is achieved by the non-adaptive query procedure in Algorithm \ref{alg:1}. Analogously to~\cite{zhou2023resolution}, we first partition the $d$-dimensional unit cube $[0,1]^d$ into $(nM)^d$ equal-sized non-overlapping sub-cubes $(\mathcal{C}_1,\ldots,\mathcal{C}_{(nM)^d})$ to quantize the trajectory of $k$ targets. The reason for choosing $nM$ instead of $M$ is to ensure that the trajectory estimation resolution is within $\frac{1}{M}$ within $n$ time slots. The query procedure is divided into two phases: i) query generation and response collection, and ii) decoding. In the first phase, we choose random coding in~\cite{polyanskiy2011feedback} as the encoding method. All $n$ queries $\calA^n$ are generated based using $n$ randomly generated binary vectors $(x^n(1),\ldots,x^n((nM)^d))$ with length-$(nM)^d$. All $n$ noiseless responses from the oracle are corrupted by the query-dependent noisy channel, which yields noisy responses $y^n$. In the decoding phase, the decoder uses a single threshold rule inspired by~\cite[Section IV]{yavas2023variable}, where it is chosen as a decoding rule for variable-length sparse feedback code under MAC. Compared to the query generation phase of the algorithm in~\cite[Section III-B]{zhou2022resolution}, the unit cube is divided into $nM$ sub-cubes instead of $M$ sub-cubes in line $3$ of Algorithm $1$, meaning that searching region consists of all possible motion trajectories instead of all possible locations. Compared to the decoding phase of the algorithm in~\cite[Section III-B]{zhou2022resolution} that has $2^k-1$ thresholds for $k$ targets, our single threshold rule in \eqref{single_thres} of Algorithm \ref{alg:1} significantly reduces the computational complexity of the query procedure. Compared to the decoding phase of the algorithm in~\cite[Section III-B]{zhou2023resolution}, we use the information spectrum method~\cite[Lemma 7.10.1]{koga2002information} with a single threshold instead of the method finding the maximum mutual information density.

We next briefly explain what each term in Theorem \ref{th1} means. In \eqref{th1-rlt}, the first term $4n\exp(-2n^dM^d\eta^2)$ and the multiplicative term $\exp(n \eta Lc(h(p)))$ result from the atypicality of the query-dependent channel and the change-of-measure technique, respectively, analogously to~\cite{zhou2023resolution}. The other three terms inside the parenthesis result from the three error events leading to excess-resolution. Note that the true trajectory set consists of $k$ distinct trajectories. The first term upper bounds the probability that the mutual information density of the true trajectory set is smaller than the threshold $\lambda$. The second term upper bounds the probability that there exists another trajectory set that passes the threshold decoding rule and has no intersection with the true trajectory set. The last term upper bounds the probability that there exists another trajectory set that passes the threshold decoding rule and has intersection of less than $k$ trajectories compared to the true trajectory set. Theorem \ref{th1} is similar to~\cite[Theorem 1]{zhou2022resolution} when maximum speed $v_+=0$, with the differences coming from the improvement of single threshold decoding rule and no need to account for overlapping of two trajectories.

Since the query-dependent AWGN channel violates the continuous assumption in \eqref{dmc}, Theorem \ref{th1} does not apply. The problem can be solved by controlling the tail probability of Gaussian noise as in~\cite{zhou2023resolution}. The following definition is needed to present the result for Gaussian noise. Given any $(n,p,\sigma,\eta)\in\bbN \times (0,1) \times \bbR_+^2$, define
\begin{align}
\Xi(n,p,\sigma,\eta):=\frac{nL\eta}{h(p)-L\eta} + \frac{nL\eta\big(h(p)^2+L\eta(2h(p)+L\eta)\big)(2h(p)+L\eta)}{h(p)^2\big(h(p)^2-L\eta(2h(p)+L\eta)\big)}. \label{aleph}
\end{align}

\begin{theorem} \label{th2}
Given any positive real number $\sigma\in\bbR_+$, there exists an $\big(n,k,d,\frac{2}{M},\varepsilon\big)$-non-adaptive query procedure such that
\begin{align}
\varepsilon 
\nn&\leq 4n\exp\big(-2n^dM^d\eta^2\big) + \exp\Big(-\frac{n(1-\log 2)}{2}\Big)\\
&\qquad+ \exp\big(\Xi(n,p,\sigma,\eta)\big)\bigg( \underset{(P_{X_{[k]} Y}^{h(p)})^n}{\mathrm{Pr}}\Big\{ \calG_1(n,\gamma) \Big\} + \bbH_1\Big(\big((2nv_++3)n^4M^2\big)^d-k,k,\gamma\Big) \bigg. \nn\\
&\qquad\bigg. +\sum_{\calJ\in\Lambda([k])} \bigg( \underset{(P_{X_{[k]} Y}^{h(p)})^n}{\mathrm{Pr}}\Big\{\calG_2(n,\calJ,\lambda_{\calJ})\Big\} + \bbH_1\Big(\big((2nv_++3)n^4M^2\big)^d-k,k-|\calJ|,\gamma-nC(p,|\calJ|)-n\lambda_\calJ\Big) \bigg) \bigg). \label{th2-rlt} 
\end{align}
\end{theorem}

The proof of Theorem \ref{th2} is analogous to that of Theorem \ref{th1} and the important differences are emphasized in Section \ref{Achievable Non-Asymptotic Bound AWGN}. Relative to Theorem \ref{th1}, there is an additional term $\exp\big(-\frac{n(1-\log 2)}{2}\big)$ and coefficient term $\Xi(n,p,\delta,\eta)$, which quantify the effect of truncating the channel output.

The next result provides a non-asymptotic converse bound on the performance of any optimal non-adaptive query procedure. Fix any $(n,k,d,\varepsilon)\in\bbN^3\times(0,1)$ and recall the definition of $\imath_{\calA}(\cdot)$ in \eqref{mu_info_a}.

\begin{theorem} \label{th3}
Given any $\beta\in\big(0,\frac{1-\varepsilon}{2}\big)$ and any $\kappa\in(0,1-\varepsilon-2(1+4v_+)k^2d\beta)$, any $\delta(n,k,d,\varepsilon)$-non-adaptive query procedure satisfies
\begin{align}
-\log\delta
\nn&\leq \sup_{\calA^n\in([0,1]^d)^n}\frac{1}{2dk}\Bigg(\sup\Bigg\{r\in\bbR_+:~\mathrm{Pr}\Bigg\{\sum_{i\in[n]}\imath_{\calA_i}(Z_i;Y_i) \leq r\Bigg\} \leq \varepsilon+2(1+4v_+)k^2d\beta+\kappa\Bigg\}\\*
&\qquad\qquad\qquad\qquad\qquad\qquad\qquad\qquad\qquad\qquad\qquad\qquad\qquad\qquad\qquad-dk\log(2nv_+\beta^2)-\log\kappa\Bigg), \label{th3-rlt} 
\end{align}
where $Z_i\sim\mathrm{Bern}(1-(1-|\calA_i|)^k)$.
\end{theorem}
The proof of Theorem \ref{th3} is similar to~\cite{zhou2021resolution,zhou2022resolution,zhou2023resolution} and is provided in Section \ref{Non-Asymptotic Converse Bound} for completeness. Theorem \ref{th3} resembles~\cite[Theorem 2]{zhou2022resolution} and~\cite[Theorem 3]{zhou2023resolution}, reflecting the similarity of analysis approach, specifically, converting the performance analysis of an non-adaptive query procedure to the performance analysis for a channel coding problem, followed by the application of the non-asymptotic converse bound~\cite[Proposition 4.4]{tan2014asymptotic} in channel coding.

\subsection{Second-Order Asymptotic Bounds} \label{mr Second-Order Asymptotic Bounds}

Recall the definitions of $C(k)$ in \eqref{ck} and $V(k,\varepsilon)$ in \eqref{vk}, respectively. Using the non-asymptotic bounds in Theorems \ref{th1} to \ref{th3}, we obtain the following second-order asymptotic approximation for any query-dependent DMC satisfying \eqref{dmc} and any query-dependent AWGN channel.

\begin{theorem} \label{th4}
For any excess-resolution probability $\varepsilon\in(0,1)$, the minimal achievable resolution $\delta^*(n,k,d,\varepsilon)$ of an optimal non-adaptive query procedure satisfies
\begin{align}
-\log \delta^*(n,k,d,\varepsilon)&\geq \frac{nC(k)+\sqrt{nV(k,\varepsilon)}\Phi^{-1}(\varepsilon)-4dk\log n-nv_+ + O(1)}{2dk},\label{th4-rlt}\\
-\log \delta^*(n,k,d,\varepsilon)&\leq \frac{nC(k)+\sqrt{nV(k,\varepsilon)}\Phi^{-1}(\varepsilon)+O(\log n)}{2dk}. \label{th4-rlt-2}
\end{align}
\end{theorem}

The proof of Theorem \ref{th4} requires an extension of the proofs in~\cite[Theorem 4]{zhou2023resolution} and~\cite[Theorem 3]{zhou2022resolution}. The lower bound \eqref{th4-rlt} and the upper bound \eqref{th4-rlt-2} on the negative log resolution become identical in the asymptotic limit as $n \rightarrow \infty$. The proposed single threshold procedure in Algorithm \ref{alg:1} asymptotically achieves the bound \eqref{th4-rlt}, and is therefore resolution optimal in the limit of large $n$. Theorem \ref{th4} implies that the minimal achievable resolution $\delta^*(n,k,d,\varepsilon)$ is closely related to the maximum speed $v_+$. Specifically,
\begin{align}
-2dk\log\delta^*(n,k,d,\varepsilon) =\left\{
\begin{aligned}
&nC(k)+O(nv_+)&\mathrm{if}~&nv_+=O(n^{\nu}),\nu\in[0.5,1], \\
&nC(k)+\sqrt{nV(k,\varepsilon)}\Phi^{-1}(\varepsilon)+O(nv_+)&\mathrm{if}~&nv_+=O(n^{\nu}),\nu\in(0,0.5), \\
&nC(k)+\sqrt{nV(k,\varepsilon)}\Phi^{-1}(\varepsilon)+O(\log n)&\mathrm{if}~&nv_+=O(1),
\end{aligned}
\right. \label{vclass}
\end{align}
where \eqref{vclass} identifies different regimes of $v_+$ as a function of the number of queries $n$. The above result generalizes~\cite[Corollary 2]{zhou2023resolution} that applies only to a single moving target.

Compared to~\cite[Theorem 3]{zhou2022resolution} for multiple stationary target search, Theorem \ref{th4} differs in two aspects: i) the denominator has an extra coefficient of $2$, ii) the numerator has an extra penalty term $nv_+$. The first term reflects that searching for moving targets twice as hard as searching for stationary targets. Put another way, searching for moving targets is equivalent to searching for double dimensional stationary targets, which specializes to~\cite[Theorem 4]{zhou2023resolution} when $k=1$. The second term appears since moving target search has maximal speed $v_+$. 

When specialized to stationary target search with $v_+=0$, the query procedure in Algorithm \ref{alg:1} with single threshold decoding achieves the same performance as the optimal query procedure~\cite{zhou2022resolution} that has $2^k-1$ thresholds for $k$ targets. Since only a single threshold is used, the computational complexity of Algorithm \ref{alg:1} is much lower.

Finally, in Fig. \ref{fig2}, we plot the second-order approximations to achievable resolutions as a function of the number of targets for various cases of maximal speed $v_+$. As observed, when the number of targets $k$ or the maximal speed $v_+$ increases, the achievable resolution decreases. The above numerical result is consistent with our intuition since it is more challenging to locate more number of targets or faster moving targets.

\begin{figure}[tb]
\centering
\includegraphics[width=0.5\columnwidth]{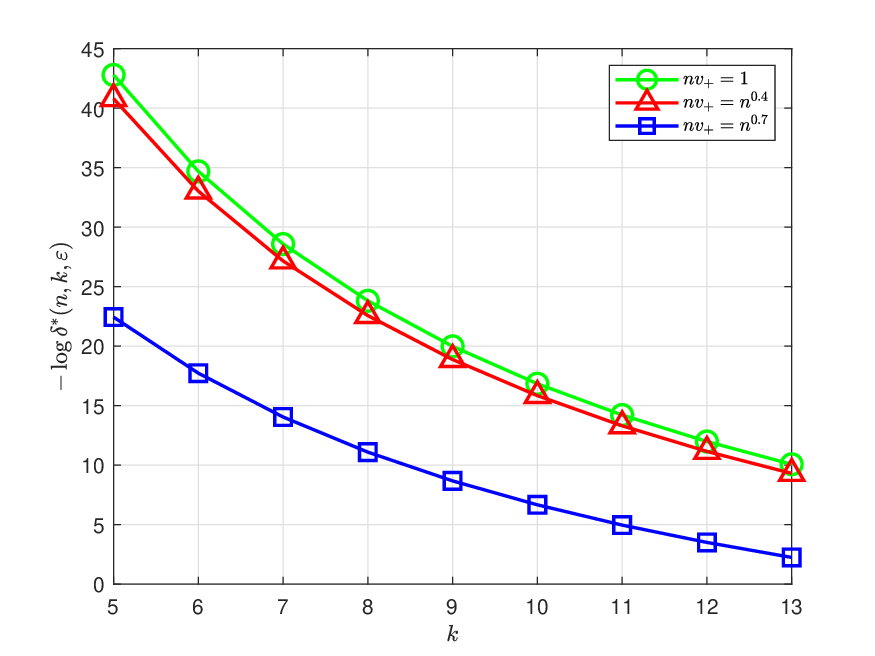}
\caption{Plot of achievable resolution of optimal non-adaptive query procedures for multiple moving target search with excess-resolution probability $\varepsilon = 0.3$ under a query-dependent BSC with a Lipschitz continuous function $h(p)=0.1+0.3p$ when the number of queries $n=1000$ and the dimension $d=1$.}
\label{fig2}
\end{figure}

\section{Effect of prior knowledge and an extension} \label{sec_sg} 

In this section, we discuss two specializations and one generalization of $k$ moving target search. Two specializations are the case with unknown initial locations but known velocities and the case known initial locations and unknown velocities, respectively. For the former case, we show that the theoretical benchmark is exactly the same as stationary target search in~\cite{zhou2022resolution}. For the latter case, the result is analogous to Theorem \ref{th4} with the exception that the factor $2$ in the denominator disappears. This is because the latter case corresponds to the search over the velocity only, which generalizes the analysis of a single moving target~\cite{zhou2023resolution} with $B\geq 2$. Combining the latter specialization and Theorem \ref{th4}, we study the extension to a piecewise constant velocity model with $B$ time slots as in~\cite{zhou2023resolution}. For simplicity, in all settings, we only emphasize the difference of query-procedure and present the theoretical benchmarks.

\subsection{Prior Information on Target Velocities} \label{sg Unknown Initial Locations and Known Velocities Case}

Let $\bv^k=(\bv_1,\ldots,\bv_k)\in\calV^k$ be the known velocities of $k$ targets. A non-adaptive query procedure is defined as follows.

\begin{definition} \label{def4}
Given any three positive integers $(n,k,d)\in\bbN^3$ and two positive real numbers $(\delta,\varepsilon)\in\bbR_+\times(0,1)$, an $(n,k,d,\delta,\varepsilon)$-non-adaptive query procedure consists of
\begin{itemize}
\item[$\bullet$] $n$ query sets $(\calA_1,\dots,\calA_n) \subseteq ([0,1]^d)^n$,
\item[$\bullet$] a decoding function $g:\calY^n \rightarrow \calS^k$,
\end{itemize}
such that the excess-resolution probability satisfies
\begin{align}
\mathrm{P_e}(n,k,d,\delta):=\sup_{f_{\bS^k}\in\calF(\calS^k)} \mathrm{Pr}\bigg\{ \exists~t\in[k]:\min_{\hatbs\in\hatbs^k} \max_{i\in[0:n]} \big\| l\big(\hatbs,\bv_t,i\big)-l\big(\bS_t,\bv_t,i\big) \big\|_{\infty} > \delta \bigg\} \leq \varepsilon ,
\end{align}
where $\hatbs^k = \{\hatbs_1,\ldots,\hatbs_k\}$ is the set of estimated initial locations.
\end{definition}

Given any $s\in[0,1]$, define a quantization function
\begin{align}
\bar{\mathrm{q}}(s):=\lceil sM \rceil.
\end{align}
Analogously to \eqref{quantized-loc}, the quantized location is defined as
\begin{align}
\bar{w}(s_j,v_j,i):=\mathrm{q}\big(l(s_j,v_j,i)\big) ,
\end{align}
and the quantized location of each target is a $d$-dimensional vector $\bar{w}(\mathbf{s},\mathbf{v},i)=(\bar{w}(s_1,v_1,i),\ldots,\bar{w}(s_d,v_d,i))$. Let $\bar{\bw}(\mathbf{s},\mathbf{v},[n])=(\bar{w}(\mathbf{s},\mathbf{v},1),\dots,\bar{w}(\mathbf{s},\mathbf{v},n))\in([M]^d)^n$ be the collection of quantized locations at discrete time points $[n]$ of a target. Define the set of trajectories with known velocities $\bv\in\calV$ as
\begin{align}
\calB_{n,M}(\bv):=\big\{\bar{\bw}^n\in([M]^d)^n: \bar{\bw}^n=\bar{\bw}(\bs,\bv,[n]) \mathrm{~for~some}~\bs\in[0,1]^d\big\}.
\end{align}
Given a fixed initial location $\bs\in\calS$ and a velocity $\bv\in\calV$, the quantized trajectory is fixed. Thus, 
\begin{align}
\big|\calB_{n,M}(\bv)\big|&\leq M^d\label{bnmv}.
\end{align}
In other words, when the moving velocity is known, searching for a moving target is equivalent to searching for a stationary target, with the only exception that the search should be done with quantized trajectories instead of initial locations. Compared with the original problem where both the initial locations and velocities are known, the number of variables to be estimated reduces from $2d$ to $d$ and the number of possible trajectories reduces from $((2nv_++3)n^4M^2)^d$ to $M^d$. Correspondingly, the non-adaptive query procedure in Algorithm \ref{alg:1} is simplified to Algorithm \ref{alg:2} provided in Appendix \ref{Specialization Bound 1}.

The theoretical benchmark is as follows.

\begin{theorem} \label{caseA-2order}
For any $(n, k, d, \varepsilon)\in\bbN^3 \times (0,1)$, the minimal achievable resolution $\delta^*(n,k,d,\varepsilon)$ of an optimal non-adaptive query procedure satisfies
\begin{align}
-\log \delta^*(n,k,d,\varepsilon) = \frac{nC(k)+\sqrt{nV(k,\varepsilon)}\Phi^{-1}(\varepsilon)+\Theta(\log n)}{dk}. \label{caseA-2order-rlt}
\end{align}
\end{theorem}

A proof sketch of Theorem \ref{caseA-2order} is provided in Appendix \ref{Specialization Bound 1}. The theoretical benchmark in Theorem \ref{caseA-2order} is the same as multiple stationary target search~\cite[Theorem 3]{zhou2022resolution}, indicating that the multiple moving target search is equivalent to multiple stationary targets case when targets’ velocities are known. The single threshold decoding rule in Algorithm \ref{alg:2} can be used in the stationary target search in~\cite{zhou2022resolution} to reduce the computational complexity. In Section \ref{app ft}, we have taken this case into consideration in the beam tracking procedure and numerically simulated the performance in Fig. \ref{fig56}, which shows that our theoretical result in Theorem \ref{caseA-2order} provides a good approximation to the non-asymptotic performance.

\subsection{Prior Information on Target Initial Locations} \label{sg Known Initial Locations and Unknown Velocities Case}

Let $\bs^k=(\bs_1,\ldots,\bs_k)\in\calS^k$ be the known initial locations of $k$ targets. A non-adaptive query procedure is defined as follows.

\begin{definition} \label{def5}
Given any positive integers $(n,k,d)\in\bbN^3$, any tolerable resolution $\delta\in\bbR_+$ and constant $\varepsilon\in(0, 1)$, the second special $(n,k,d,\delta,\varepsilon)$-non-adaptive query procedure consists of
\begin{itemize}
\item[$\bullet$] $n$ query sets $(\calA_1,\dots,\calA_n) \subseteq ([0,1]^d)^n$,
\item[$\bullet$] a decoding function $g:\calY^n \rightarrow \calV^k$,
\end{itemize}
such that the excess-resolution probability satisfies
\begin{align}
\mathrm{P_e}(n,k,d,\delta):=\sup_{f_{\bV^k}\in\calF(\calV^k)} \mathrm{Pr}\bigg\{ \exists~t\in[k]:\min_{\hatbv\in\hatbv^k} \max_{i\in[0:n]} \big\| l\big(\bs_t,\hatbv,i\big)-l\big(\bs_t,\bV_t,i\big) \big\|_{\infty} > \delta \bigg\} \leq \varepsilon ,
\end{align}
where $\hatbv^k = \{\hatbv_1,\ldots,\hatbv_k\}$ is the set of estimated velocities of the targets.
\end{definition}

Recall the definition of the quantized location of a moving target in \eqref{quantized-loc}. Define the set of trajectories with known initial locations $\bs\in\calS$ as
\begin{align}
\calB_{n,M}(\bs):=\big\{\bw^n\in([M]^d)^n: \bw^n=\bw(\bs,\bv,[n]) \mathrm{~for~some}~\bv\in[-v_+,v_+]^d\big\}.
\end{align}
Analogously to \eqref{bnm}, the size of $\calB_{n,M}(\bs)$ is upper bounded by
\begin{align}
\big|\calB_{n,M}(\bs)\big| \leq \big((2nv_++3)n^3M\big)^d.
\end{align}

Compared to the setting of Algorithm \ref{alg:1}, the setting of known initial locations also changes the number of unknown random variable to be estimated from $2d$ to $d$, and thus reduces the number of possible trajectories from $((2nv_++3)n^4M^2)^d$ to $((2nv_++3)n^3M)^d$. Correspondingly, the non-adaptive query procedure is summarized in Algorithm \ref{alg:3} provided in Appendix \ref{Specialization Bound 2} and the second-order asymptotic analysis is summarized in the following theorem.

\begin{theorem} \label{caseB-2order}
For any $(n, k, d, \varepsilon)\in\bbN^3 \times (0,1)$, the minimal achievable resolution $\delta^*(n,k,d,\varepsilon)$ of an optimal non-adaptive query procedure satisfies
\begin{align}
-\log \delta^*(n,k,d,\varepsilon)&\geq \frac{nC(k)+\sqrt{nV(k,\varepsilon)}\Phi^{-1}(\varepsilon)-3dk\log n - nv_+ + O(1)}{dk},\label{caseB-2order-rlt-1}\\
-dk\log \delta^*(n,k,d,\varepsilon)&\leq nC(k)+\sqrt{nV(k,\varepsilon)}\Phi^{-1}(\varepsilon)+O(\log n). \label{caseB-2order-rlt-2}
\end{align}
\end{theorem}
A proof sketch of Theorem \ref{caseB-2order} is provided in Appendix \ref{Specialization Bound 2}. The result in Theorem \ref{caseB-2order} shows that the multiple moving target case with known targets' initial locations is equivalent to the multiple targets expansion of~\cite{zhou2023resolution} at the second time slot, which differs from Theorem \ref{th4} in that the coefficient $2$ reduces to $1$. This is because the number of unknown variables to be estimated reduces from $2d$ in the case of Algorithm \ref{alg:1} to $d$ in Algorithm \ref{alg:3}.

\subsection{Extension to Piecewise Constant Velocity Model} \label{sg Piecewise Constant Velocity Model}

Here we consider the piecewise constant velocity model~\cite{zhou2023resolution} with $B$ time slots $\bn=(n_1,\ldots,n_B)$. Let $N_1=n_1$ and let $N_j:=n_j-n_{j-1}$ for each $j\in[2:B]$. Fix any $t\in[k]$. Recall that $\calS=[0,1]^d$ and $\calV=[-v_+,v_+]^d$. In the first time slot with $i\in[n_1]$, the $t$-th target moves with unknown initial location $\bS_t\in\calS$ and velocity $\bV_{t,1}\in\calV$. For each $j\in[2:B]$, in the $j$-th time slot with $i\in[n_{j-1}+1,n_j]$, each target changes their velocity to $\bV_{t,j}$. The task in this case is to estimate the trajectory of all $k$ targets of $n_B$ time slots with unknown initial locations $\bS^k=(\bS_1,\ldots,\bS_k)\in\calS^k$ and unknown velocities $\bV^{kB}=\{(\bV_{j,1},\ldots,\bV_{j,k})\}_{j\in[B]}$. 
Correspondingly, the non-adaptive query procedure is defined as follows.

\begin{definition} \label{def6}
Given any $(\mathbf{n},k,d)\in\bbN^{B+2}$, any tolerable resolution $\delta\in\bbR_+$ and constant $\varepsilon\in(0, 1)$, an $(\mathbf{n},k,d,\delta,\varepsilon)$-non-adaptive query procedure under a piecewise constant velocity model consists of
\begin{itemize}
\item[$\bullet$] $n_B$ query sets $(\calA_1,\dots,\calA_{n_B}) \subseteq ([0,1]^d)^{n_B}$,
\item[$\bullet$] a decoding function $g_1:\calY^{N_1} \rightarrow \calS^k \times \calV^k \subseteq ([0,1]^d \times [-v_+,v_+]^{d})^k$,
\item[$\bullet$] $B-1$ decoding functions $g_j:\calY^{N_j} \rightarrow \calV^k \subseteq ([-v_+,v_+]^d)^k$,
\end{itemize}
such that the excess-resolution probability satisfies
\begin{align}
&\mathrm{P_e}(\mathbf{n},k,d,\delta):=\sup_{f_{\bS^k\bV^{kB}}\in\calF(\calS^k\times\calV^k)} \mathrm{Pr}\bigg\{ \exists~t\in[k]:\min_{(\hatbs,\hatbv^B)\in(\hatbs,\hatbv^B)^k} \max_{i\in[0:n_1]} \big\| l(\hatbs,\hatbv_1,i)-l(\bS_t,\bV_{t,1},i) \big\|_{\infty} > \delta \bigg.\nn\\
&\qquad\qquad\qquad\qquad\qquad\qquad\qquad\qquad\bigg. \mathrm{or}~\min_{(\hatbs,\hatbv^B)\in(\hatbs,\hatbv^B)^k}\max_{j\in[2:B]}\max_{i\in[n_{j-1}:n_j]} \big\| l(\hatbs,\hatbv^j,i)-l(\bS_t,\bV_{t,j},i) \big\|_{\infty} > \delta \bigg\} \leq \varepsilon ,
\end{align}
where $(\hatbs,\hatbv^B)^k = \{(\hatbs_1,\hatbv_1^B),\ldots,(\hatbs_k,\hatbv_k^B)\}\in(\calS\times\calV^B)^k$ is the set of estimated initial location and velocity pairs of $k$ targets. 
\end{definition}

A non-adaptive query procedure for the piecewise constant velocity model can be constructed by using Algorithm \ref{alg:1} for the first time slot and repeatedly applying Algorithm \ref{alg:3} in the remaining $B-1$ time slots. The corresponding algorithm is summarized in Algorithm \ref{alg:4} provided in Appendix \ref{Generalization Bound}. The second-order asymptotics of optimal non-adaptive query procedure is stated in the following theorem.

\begin{theorem} \label{caseC-2order}
For any $(\mathbf{n}, k, d, \varepsilon)\in\bbN^{B+2} \times (0,1)$, the minimal achievable resolution $\delta^*(\mathbf{n},k,d,\varepsilon)$ of an optimal non-adaptive query procedure satisfies
\begin{align}
-dk\log \delta^*(\mathbf{n},k,d,\varepsilon) &\geq \max_{\substack{(\varepsilon_1,\ldots,\varepsilon_B) \\ \sum_{j\in[B]}\varepsilon_j \leq \varepsilon}} \min\Bigg\{ \frac{N_1C(k)+\sqrt{N_1V(k,\varepsilon_1)}\Phi^{-1}(\varepsilon_1)-4dk\log N_1-N_1v_+ + O(1)}{2}, \Bigg. \nn\\*
&\Bigg. \min_{j\in[2:B]} \bigg\{N_jC(k)+\sqrt{N_jV(k,\varepsilon_j)}\Phi^{-1}(\varepsilon_j)-3dk\log N_j - N_jv_+ + O(1)\bigg\} \Bigg\} - dk\log(B+1),\\
-(B+1)dk\log \delta^*(\mathbf{n},k,d,\varepsilon)&\leq n_BC(k)+\sqrt{n_BV(k,\varepsilon)}\Phi^{-1}(\varepsilon)+O(\log n_B).
\end{align}
\end{theorem}
A proof sketch of Theorem \ref{caseC-2order} is provided in Appendix \ref{Generalization Bound}. Analogously to the single moving target search in~\cite[Theorem 4]{zhou2023resolution}, the upper and lower bounds on the minimal resolution generally do not match except a few special cases, including the case of $B=1$ in Theorem \ref{th4}.

\section{Application to Beam Tracking} \label{sec_app}

Twenty questions estimation is widely studied due to its diverse applications, e.g., beam alignment and tracking in wireless communications~\cite{chiu2019active,gau2024beam,ronquillo2023integrated}, target localization using sensor networks~\cite{tsiligkaridis2014collaborative,tsiligkaridis2015decentralized} and face localization in images~\cite{rajan2015bayesian}. In this section, we use beam tracking as an example to demonstrate how our theoretical results lead to theoretical benchmarks for practical problems in wireless communications.

\subsection{Background}

In 5G NR system, beam management is used to acquire and maintain beams, enabling both the transmitter (TX) and receiver (RX) to identify the appropriate beams that are likely to contain targets at any given time point~\cite{li2020beam}. Beam management procedures are applicable for both idle mode and connected mode. The idle mode occurs when the user equipment (UE) does not have active data transmission while the connected mode occurs when active data exchange is taking place between the UE and gNodeB (gNB) and the UE is moving. One component of beam management procedures is beam sweeping, which uses a two-stage exhaustive search algorithm for both transmitting and receiving ends. However, exhaustive search algorithm is infeasible for large search space and for the case that search results may be incorrect. To tackle the problem, it is necessary to find alternative algorithms to locate and track beams accurately. Existing studies on beam tracking~\cite{va2016beam,zhao2017angle,gao2016fast,shaham2020extended,yang2019beam,huang20203d,wang2020demystifying,hussain2020mobility,seo2015training,ronquillo2023integrated} heavily rely on frequent beam training. In the following, using the twenty questions framework, we propose a beam tracking algorithm with only one beam training phase, which can significantly reduce the cost of pilot training in the entire communication system. We also characterize the performance of optimal non-adaptive beam tracking algorithms.

\subsection{System Model} \label{app sm}

\begin{table}[h]
\caption{Symbol Table}
\label{table_sym_bt}
\centering
\begin{tabular}{|m{1cm}<{\centering}|m{6.5cm}|m{1cm}<{\centering}|m{6.5cm}|}
\hline
Symbol & Meaning & Symbol & Meaning\\
\hline
$\mathtt{R}_1$ & Number of rows of uniform planar array antenna elements & $\mathtt{R}_2$ & Number of columns of uniform planar array antenna elements\\
\hline
$\mathtt{g}$ & Antenna spacing & $\mathtt{l}$ & Wavelength\\
\hline
$\phi_{\mathrm{az}}$ & Angle of arrival in azimuth & $\phi_{\mathrm{el}}$ & Angle of arrival in elevation\\
\hline
$\mathtt{a}$ & Steering vector & $\alpha$ & Fading coefficient\\
\hline
$\mathtt{h}$ & Small-scale channel & $\mathtt{s}_t$ & Pilot sequence of the $t$-th transmitter\\
\hline
$\mathtt{y}$ & Signal received by the receiver & $\sqrt{\mathtt{P}}$ & Combined effect of the transmit power and the large-scale fading\\
\hline
$\mathtt{n}$ & Equivalent Gaussian noise vector & $\iota$ & Power threshold\\
\hline
$\mathtt{z}$ & Noisy decision result for beam tracking & $\mathtt{w}$ & Beamforming vector\\
\hline
$\mathrm{q_{bt}}$ & Quantization function for beam tracking & $\mathtt{v}_{\mathrm{az}}$ & Velocity of an angle of arrival in azimuth\\
\hline
$\mathtt{v}_{\mathrm{el}}$ & Velocity of an angle of arrival in elevation & $\mathtt{b}$ & Real-time angle of arrival of a transmitter\\
\hline
\end{tabular}
\end{table}

The symbols used in the beam tracking procedure are summarized in Table \ref{table_sym_bt}. Consider an air-ground communication scenario where a single stationary base station (BS) serves as RX and multiple mobile unmanned aerial vehicles serve as TXs. Assume that the RX is equipped with a uniform planar array (UPA) with $\mathtt{R}_1 \times \mathtt{R}_2\in\bbN^2$ antennas, while each TX's antennas can be regarded as a single virtual antenna because of the fixed beamforming vector. At each sampling time point $i\in[n]$, the RX receives the signal from TXs by the UPA using the directional beamforming vector $\mathtt{w}(i)\in\mathbb{C}^{\mathtt{R}_1 \times \mathtt{R}_2}$. Furthermore, we assume that an air-ground communication scenario is dominated by the line-of-sight path without obstructions and reflectors.

Given the antenna spacing $\mathtt{g}\in\mathbb{R}_+$, the steering vector with the angle of arrival (AoA) in azimuth $\phi_\mathrm{az}\in[0,\pi]$ and elevation $\phi_\mathrm{el}\in\big[0,\frac{\pi}{2}\big]$ illustrated in Fig. \ref{fig3} can be described as:
\begin{align}
\mathtt{a}(\phi_\mathrm{az},\phi_\mathrm{el})&:=\sqrt{\frac{1}{\mathtt{R}_1\times\mathtt{R}_2}}\Big( 1,e^{-j\frac{2\pi \mathtt{g}}{\mathtt{l}}\sin\phi_\mathrm{az}\cos\phi_\mathrm{el}},\ldots,e^{-j(\mathtt{R}_1-1)\frac{2\pi \mathtt{g}}{\mathtt{l}}\sin\phi_\mathrm{az}\cos\phi_\mathrm{el}} \Big) \otimes \Big( 1,e^{-j\frac{2\pi \mathtt{g}}{\mathtt{l}}\sin\phi_\mathrm{el}},\ldots,e^{-j(\mathtt{R}_2-1)\frac{2\pi \mathtt{g}}{\mathtt{l}}\sin\phi_\mathrm{el}} \Big) \\
&=\sqrt{\frac{1}{\mathtt{R}_1\times\mathtt{R}_2}}\Big( 1,\ldots,e^{-j\frac{2\pi \mathtt{g}}{\mathtt{l}}(\mathtt{r}_1\sin\phi_\mathrm{az}\cos\phi_\mathrm{el}+\mathtt{r}_2\sin\phi_\mathrm{el})},\ldots,e^{-j\frac{2\pi \mathtt{g}}{\mathtt{l}}((\mathtt{R}_1-1)\sin\phi_\mathrm{az}\cos\phi_\mathrm{el}+(\mathtt{R}_2-1)\sin\phi_\mathrm{el})} \Big) ,
\end{align}
where $\mathtt{r}_1\in[0:\mathtt{R}_1-1], \mathtt{r}_2\in[0:\mathtt{R}_2-1]$, $\mathtt{l}$ denotes wavelength and $\otimes$ denotes Kronecker multiplication operation.

\begin{figure}[tb]
\centering
\includegraphics[width=0.4\columnwidth]{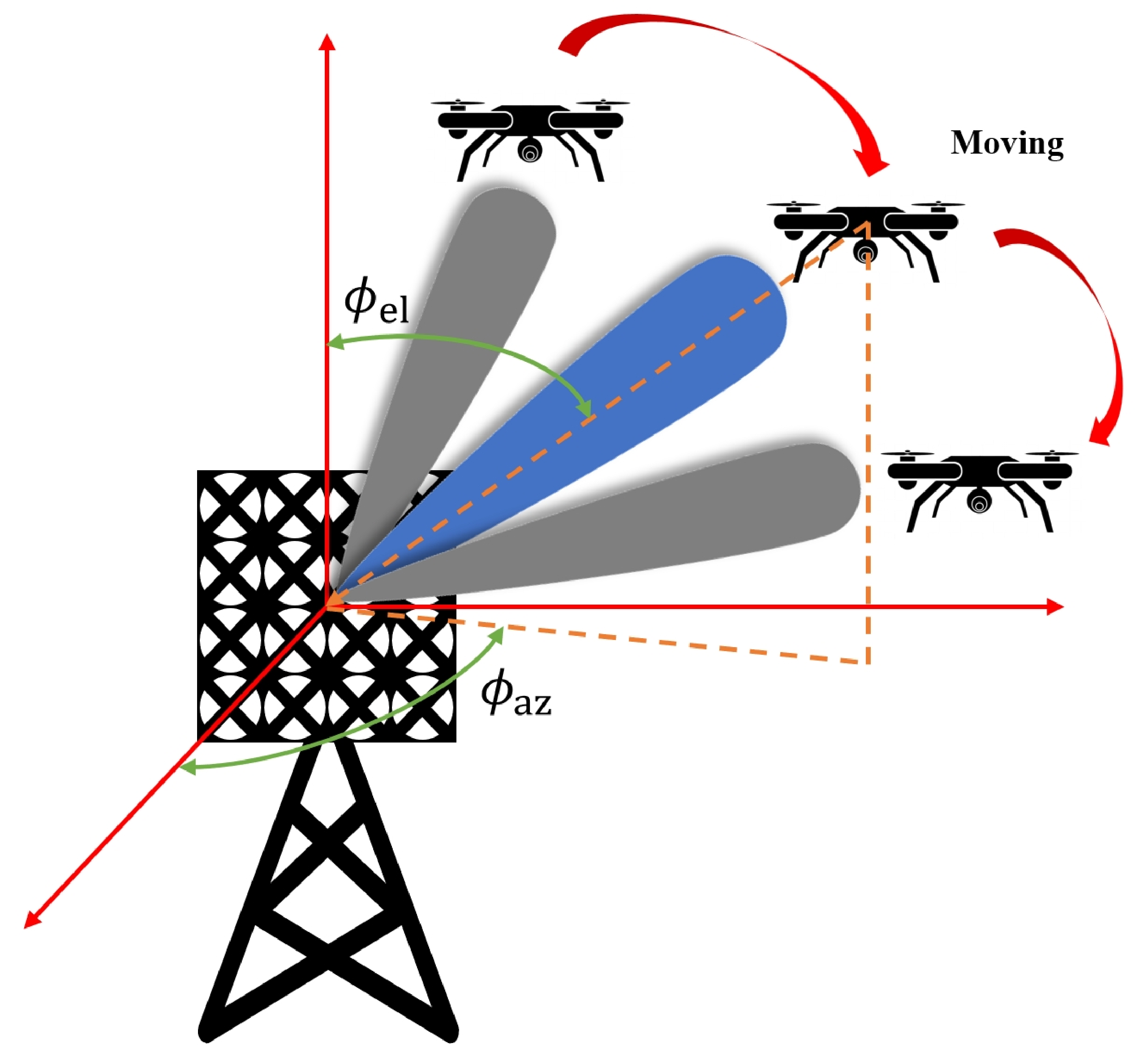}
\caption{The AoA of a beam from a TX to the RX with azimuth $\phi_\mathrm{az}$ and elevation $\phi_\mathrm{el}$.}
\label{fig3}
\end{figure}

Given the fading coefficient $\alpha\in\mathbb{C}$, the small-scale channel can be described as:
\begin{align}
\mathtt{h}:=\alpha\mathtt{a}(\phi_\mathrm{az},\phi_\mathrm{el}).
\end{align}

Assume that there are $k$ TXs in the communication system and the pilot sequence of the $t$-th TX is $\mathtt{s}_t\in\bar{\calS}$ for any $t\in[k]$. Then the received signal at the sampling time point $i$ is
\begin{align}
\mathtt{y}(i)&:=\sqrt{\mathtt{P}}\mathtt{w}^\mathrm{H}(i)\bigg(\sum_{t=1}^k \mathtt{h}_{t}\mathtt{s}_{t}\bigg)+\mathtt{w}^\mathrm{H}(i)\mathtt{n}(i) \\
&=\alpha\sqrt{\mathtt{P}}\mathtt{w}^\mathrm{H}(i)\bigg(\sum_{t=1}^k \mathtt{a}_t\big(\phi_{\mathrm{az},t},\phi_{\mathrm{el},t}\big)\mathtt{s}_{t}\bigg)+\mathtt{w}^\mathrm{H}(i)\mathtt{n}(i) ,
\end{align}
where $\sqrt{\mathtt{P}}$ denotes the combined effect of the transmit power and the large-scale fading, and $\mathtt{n}(i) \sim \mathcal{CN}(\mathbf{0}_{\mathtt{R}_1\mathtt{R}_2 \times 1},\sigma^2\mathtt{I})$ is the equivalent Gaussian noise vector.

For practical communication systems, given any power threshold $\iota\in\mathbb{R}_+$, we adopt the following $1$-bit measurement model in~\cite{chiu2019active}:
\begin{align}
\mathtt{z}(i):=\bbo\big(|\mathtt{y}(i)|^2 > \iota\big) ,
\end{align}
which means that at each sampling time point $i\in[n]$, the RX only has $1$-bit of information indicating whether or not the received power passes the threshold $\iota$.

\subsection{From Beam Tracking To Twenty Questions Estimation} \label{app ft}

\begin{figure}[tb]
\centering
\includegraphics[width=.7\columnwidth]{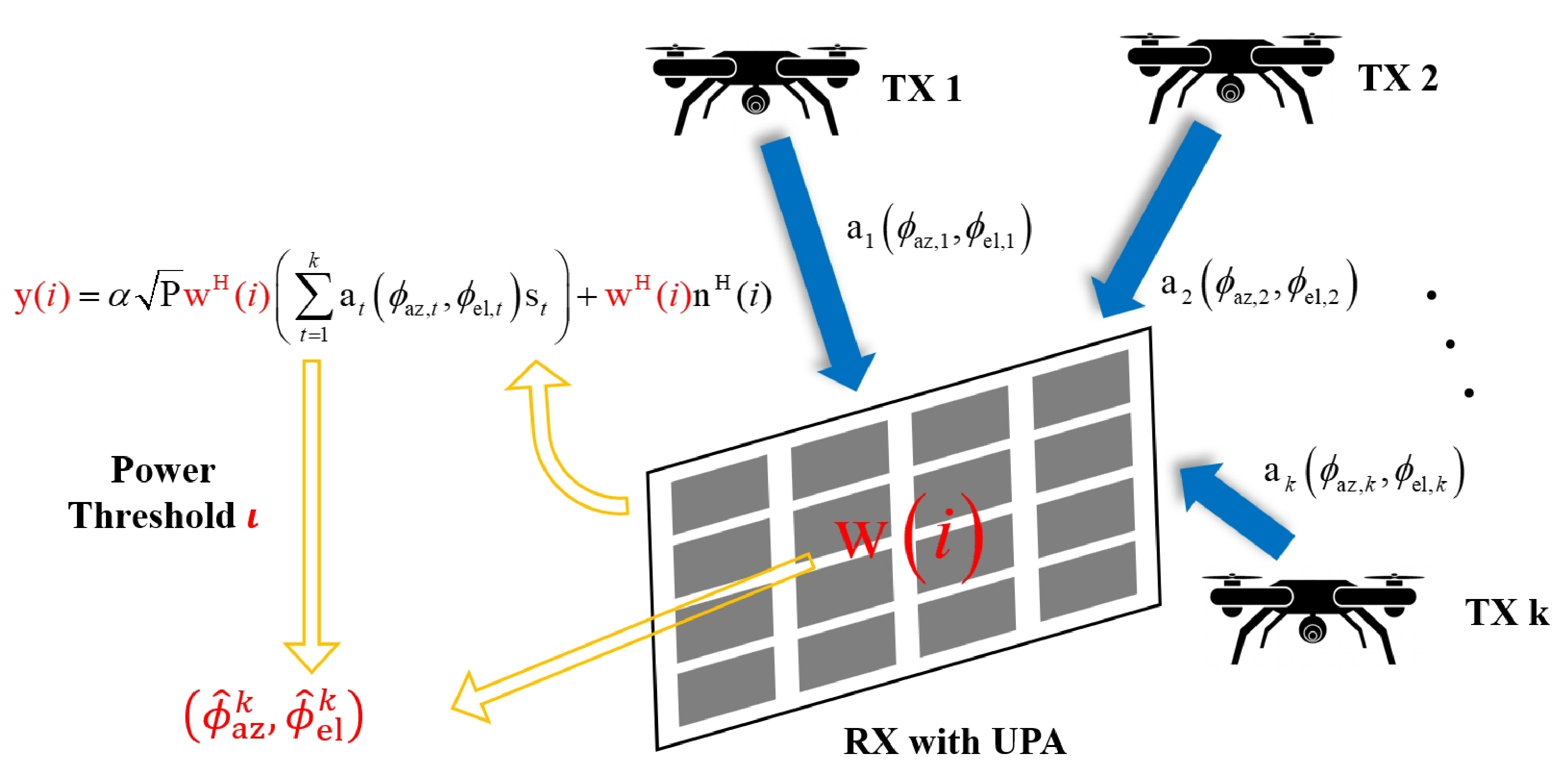}
\caption{The beam tracking procedure with $k$ mobile TXs moving over the plane ($d=2$) and a stationary RX.}
\label{fig4}
\end{figure}

The first step is to quantize the angular space $[0,\pi]\times\big[0,\frac{\pi}{2}\big]$ of the RX into $M \times \frac{M}{2}$ sub-spaces denoted as $\big(\mathtt{C}_1,\ldots,\mathtt{C}_{\frac{M^2}{2}}\big)$. Given the total number of the sampling time points $n\in\mathbb{N}$, for any real-time AoA pair $\phi_{\mathrm{az,el}}:=(\phi_\mathrm{az},\phi_\mathrm{el})\in[0,\pi]\times\big[0,\frac{\pi}{2}\big]$, define the quantization function as follows:
\begin{align}
\mathrm{q}_{\mathrm{bt}}(\phi_{\mathrm{az,el}},n):= \Big(\Big\lceil\frac{\phi_\mathrm{az}}{\pi}\times nM\Big\rceil-1\Big) \times nM+\Big\lceil\frac{\phi_\mathrm{el}}{\pi/2}\times\frac{nM}{2}\Big\rceil.
\end{align}

Given any $\phi_{\mathrm{az,el}}\in[0,\pi]\times\big[0,\frac{\pi}{2}\big]$ and $\mathtt{v}_{\mathrm{az,el}}:=(\mathtt{v}_\mathrm{az},\mathtt{v}_\mathrm{el})\in[-\mathtt{v}_+,-\mathtt{v}_+]^2$, if a beam moves with a constant velocity pair $\mathtt{v}_{\mathrm{az,el}}$ from an initial AoA pair $\phi_{\mathrm{az,el}}$ at any sampling time point $i\in[n]$, the coordinate $\mathtt{b}(\phi_{\mathrm{az,el}},\mathtt{v}_{\mathrm{az,el}},i)$ of the real-time AoA satisfies
\begin{align}
\mathtt{b}(\phi_{\mathrm{az,el}},\mathtt{v}_{\mathrm{az,el}},i):=\Big(\mathrm{mod}\big(\phi_{\mathrm{az}}+i\mathtt{v}_{\mathrm{az}},\pi\big),\mathrm{mod}\Big(\phi_{\mathrm{el}}+i\mathtt{v}_{\mathrm{el}},\frac{\pi}{2}\Big)\Big).
\end{align}
After all $n$ sampling time points, let $\mathtt{b}(\phi_{\mathrm{az,el}},\mathtt{v}_{\mathrm{az,el}},[n])$ denote $(\mathtt{b}(\phi_{\mathrm{az,el}},\mathtt{v}_{\mathrm{az,el}},1),\ldots,\mathtt{b}(\phi_{\mathrm{az,el}},\mathtt{v}_{\mathrm{az,el}},n))$ and $\mathrm{q}_{\mathrm{bt}}^n(\phi_{\mathrm{az,el}}^n,n)$ denote $(\mathrm{q}_{\mathrm{bt}}(\phi_{\{\mathrm{az,el}\},1},n),\ldots,\mathrm{q}_{\mathrm{bt}}(\phi_{\{\mathrm{az,el}\},n},n))$.

Assume that there are $k$ TXs in the communication system with initial AoA pairs $\phi_{\mathrm{az,el}}^k\in\big([0,\pi]\times\big[0,\frac{\pi}{2}\big]\big)^k$ and velocity pairs $\mathtt{v}_{\mathrm{az,el}}^k\in[-\mathtt{v}_+,-\mathtt{v}_+]^{2k}$. Recall the system model defined in Section \ref{app sm}. The beam tracking procedure is defined as follows and \blue{is} illustrated in Fig. \ref{fig4}. At first, the RX sets $n$ sampling time points and generates beamforming vectors $\mathtt{w}(i)$ for each time point $i$ in advance. Then $k$ TXs move at their respective constant velocities from their initial locations and send pilot sequences to the RX. After $n$ samples, the RX obtains $n$ decision results using power threshold $\iota$ and estimates the initial AoA pairs $\hat{\phi}_\mathrm{az,el}^k$ and velocity pairs $\hat{\mathtt{v}}_\mathrm{az,el}^k$ based on them.

\begin{definition} \label{def7}
Given any total number of sampling time points and TXs $(n,k)\in\bbN^{2}$, any tolerable resolution $\delta\in\bbR_+$ and constant $\varepsilon\in(0, 1)$, an $(n,k,\delta,\varepsilon)$ beam tracking procedure consists of
\begin{itemize}
\item[$\bullet$] $n$ query set $(\calA_1,\dots,\calA_n) \subseteq \big([0,\pi]\times\big[0,\frac{\pi}{2}\big]\big)^{n}$ and $n$ corresponding beamforming vector $(\mathtt{w}_{\calA_1}(1)\dots,\mathtt{w}_{\calA_n}(n)) \subseteq \mathbb{C}^{n}$, 
\item[$\bullet$] a power threshold $\iota\in\mathbb{R}_+$, making that for each sampling time point $i\in[n]$, the received signal power exceeds $\iota$ when at least one real-time AoA pair satisfies $\mathtt{b}(\phi_{\mathrm{az,el}},\mathtt{v}_{\mathrm{az,el}},i)\in\calA_i$ and does not exceed $\iota$ when none pair satisfies,
\item[$\bullet$] a estimator $g_\mathrm{bt}$ that obtains estimates of AoA pairs $\hat{\phi}_\mathrm{az,el}^k\in\big([0,\pi]\times\big[0,\frac{\pi}{2}\big]\big)^{k}$ and velocity pairs $\hat{\mathtt{v}}_\mathrm{az,el}^k\in[-\mathtt{v}_+,-\mathtt{v}_+]^{2k}$ after $n$ queries,
\end{itemize}
such that the excess-resolution probability satisfies
\begin{align}
&\mathrm{P_e}(n,k,\delta) \nn\\
&:= \sup \mathrm{Pr}\bigg\{ \exists~t\in[k]:\min_{(\hat{\phi}_\mathrm{az,el},\hat{\mathtt{v}}_\mathrm{az,el})\in(\hat{\phi}_\mathrm{az,el}^k,\hat{\mathtt{v}}_\mathrm{az,el}^k)} \max_{i\in[0:n]} \Big\| \mathtt{b}\big(\hat{\phi}_{\mathrm{az,el}},\hat{\mathtt{v}}_{\mathrm{az,el}},i\big)-\mathtt{b}\big(\phi_{\{\mathrm{az,el}\},t},\mathtt{v}_{\{\mathrm{az,el}\},t},i\big) \Big\|_{\infty} > \delta \bigg\} \leq \varepsilon.
\end{align}
\end{definition}

The following theorem establish the equivalence of twenty questions estimation and beam tracking.

\begin{theorem} \label{20Q-BeamTracking}
For any $(n,k,\delta,\varepsilon)\in\bbN^2 \times \bbR_+ \times [0,1]$, with any $(n,k,\delta,\varepsilon)$ beam tracking procedure, we can construct an $\big(n,k,2,\frac{\delta}{\pi},\varepsilon\big)$-non-adaptive query procedure for twenty questions estimation. Furthermore, with any $(n,k,2,\delta,\varepsilon)$-non-adaptive query procedure for twenty questions, we can construct an $(n,k,2,\pi\delta,\varepsilon)$ for beam tracking. 
\end{theorem}

\begin{proof}
Fix any $(n,k,\delta,\varepsilon)$ beam tracking procedure. By normalizing the entire angular space $[0,\pi]\times\big[0,\frac{\pi}{2}\big]$ to the two-dimensional unit cube $[0,1]^2$, we can associate the AoAs to moving targets with initial locations $\bS^k$ and velocities $\bV^k$ to be estimated. The angular space covered by $n$ beamforming vectors $(\mathtt{w}_{\calA_1}(1),\ldots,\mathtt{w}_{\calA_n}(n))$ can also be normalized to $n$ query regions $(\calA_1,\ldots,\calA_n)$. Then, we can estimate the initial locations $\hat{\bS}^k$ and velocities $\hat{\bV}^k$ of the $k$ targets after using the power threshold $\iota$ to obtain noisy query responses, and de-normalize them to the angular space.

Fix any $(n,k,2,\delta,\varepsilon)$-non-adaptive query procedure. By mapping the two-dimensional unit cube $[0,1]^2$ to the angular space $[0,\pi]\times\big[0,\frac{\pi}{2}\big]$, the $k$ moving targets represent $k$ TXs with AoA pairs $\phi_\mathrm{az,el}^k$ and velocity pairs $\mathtt{v}_\mathrm{az,el}^k$. Then, we can generate beamforming vectors $(\mathtt{w}_{\calA_1}(1),\ldots,\mathtt{w}_{\calA_n}(n))$ for searching the corresponding space based on the random query regions $(\calA_1,\ldots,\calA_n)$. Furthermore, we serve the noisy responses from the oracle as the decision results obtained by the RX using $1$-bit measurement model. Finally, we can estimate the AoA pairs $\hat{\phi}_\mathrm{az,el}^k$ and the velocity pairs $\hat{\mathtt{v}}_\mathrm{az,el}^k$ with the help of the decoder in the beam tracking procedure, and remap them to the unit cube.
\end{proof}

\begin{figure}[tb]
\centering
\includegraphics[width=0.5\columnwidth]{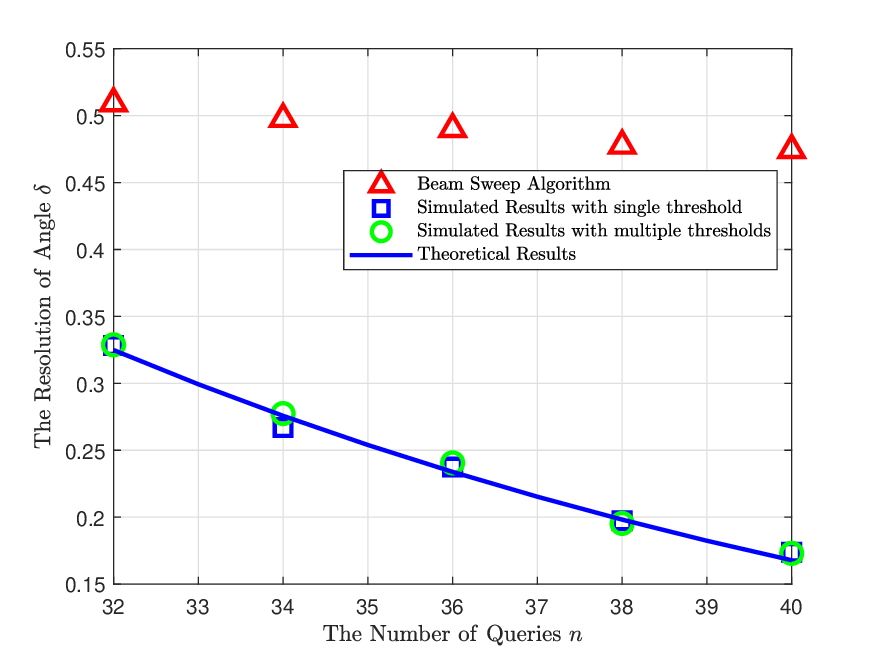}
\caption{Achievable resolutions three beam tracking procedures (red triangles, blue squares and green circles) in addition to the theoretical prediction of optimal achievable performance (blue curve). Performance of the proposed beam tracking procedure is close to the theoretically minimal achievable resolution \eqref{th4-rlt-2}. There are two TXs, the excess-resolution probability is $\varepsilon = 0.2$, and the channel is a query-dependent AWGN channel with $\sigma = 1.5$ and $h(p) = 0.4 + p$.}
\label{fig56}
\end{figure}

Thus, Theorem \ref{20Q-BeamTracking} translates the theoretical benchmarks in Theorem \ref{th4} to beam tracking. Inspired by the proof of Theorem \ref{20Q-BeamTracking}, we design an algorithm (Algorithm \ref{alg:5}) for beam tracking. Furthermore, to demonstrate the validity of Theorem \ref{20Q-BeamTracking}, in Fig. \ref{fig56}, we plot the second-order approximations to achievable resolutions of theoretical results and simulated results in blue as a function of the number of queries $n$ when $k=2$ and $d=1$. The non-adaptive procedure in Algorithm \ref{alg:5} is run independently $2 \times 10^3$ times, where AoA and velocity for each device are generated independently from uniform distributions. We assume that the channel is a query-dependent AWGN channel with $\sigma=1.5$ and fix the query size function in Definition \ref{def2} to the affine function $h(p)=0.4+p$. Such a choice is natural since it aligns with practical applications where the additive noise accumulates as the query region expands~\cite{chiu2019active}. We also simulate the achievable resolution of the beam sweep algorithm proposed in the 5G beam management standard~\cite{ZTE2017OnCSI} (plotted as red triangles in Fig. \ref{fig56}). Specifically, we fix the number of beams to be traversed in the coarse searching stage to $8$, while increasing the number of beams to be traversed in the fine searching stage according to the number of queries $n$, i.e., the number of beams to be traversed in the fine searching stage is fixed to $\frac{n-8}{k}$ within each angular region estimated by the coarse searching stage. The remaining settings are consistent with the proposed beam tracking algorithm in Algorithm \ref{alg:5}. Finally, for comparison, we plot the simulated results of the multiple threshold scheme proposed in~\cite[Algorithm 1]{zhou2022resolution} (denoted by green circles).

\begin{algorithm}[tb]
\caption{Beam Tracking Procedure Using Noisy 20 Questions Search}
\label{alg:5} 
\begin{algorithmic}[0]  
\REQUIRE
The number of sampling time points $n\in\bbN$, the power threshold $\iota\in\bbR_+$ and three design parameters $(M,p,\gamma)\in\bbN \times (0,1) \times \bbR_+$
\ENSURE
The sets of estimated initial AoA pairs $\hat{\phi}_\mathrm{az,el}^k\in\big([0,\pi]\times\big[0,\frac{\pi}{2}\big]\big)^k$ and velocity pairs $\hat{\mathtt{v}}_\mathrm{az,el}^k\in[-\mathtt{v}_+,-\mathtt{v}_+]^{2k}$
~\\
\hrule 
~\\
\STATE Divide the angular space $[0,\pi]\times\big[0,\frac{\pi}{2}\big]$ into $\frac{n^2M^2}{2}$ equally sized non-overlapping sub-spaces $\big(\calC_1,\ldots,\calC_{\frac{n^2M^2}{2}}\big)$.
\STATE \textbf{\emph{Beamforming vector generation and signal sampling:}}
\STATE Generate $n$ binary vectors $\big(x^n(1),\ldots,x^n\big(\frac{n^2M^2}{2}\big)\big)$ with length-$\frac{n^2M^2}{2}$, where each vector is generated i.i.d. from Bern$(p)$.
\FOR{$i\in[n]$}
\STATE Form a query region $\calA_i$ as
\begin{align}
\calA_i:=\bigcup_{j\in[\frac{n^2M^2}{2}]: x_{i,j}=1} \calC_j \nn.
\end{align}
\STATE Generate corresponding beamforming vector $\mathtt{w}(i)$ based on the $i$-th region $\calA_i$.
\STATE Take samples and obtain a noisy decision result $y_i$ on whether the received power $|\mathtt{y}(i)|^2$ pass the threshold $\iota$.
\ENDFOR
\STATE Collect noisy decision results $y^n=(y_1,\ldots,y_n)$.
\STATE \textbf{\emph{Decoding:}}
\IF{there exists sets of estimated initial AoA pairs $\hat{\phi}_\mathrm{az,el}^k$ and velocity pairs $\hat{\mathtt{v}}_\mathrm{az,el}^k$ such that $$\imath^{h(p)}\Big(x^n\big( \mathrm{q}_{\mathrm{bt}}^n\big(\mathtt{b}\big(\hat{\phi}_{\{\mathrm{az,el}\},1},\hat{\mathtt{v}}_{\{\mathrm{az,el}\},1},[n]\big),n\big) \big),\ldots,x^n\big( \mathrm{q}_{\mathrm{bt}}^n\big(\mathtt{b}\big(\hat{\phi}_{\{\mathrm{az,el}\},k},\hat{\mathtt{v}}_{\{\mathrm{az,el}\},k},[n]\big),n\big) \big); y^n\Big) \geq \gamma,$$}
\RETURN estimates $\hat{\phi}_\mathrm{az,el}^k=\big(\hat{\phi}_{\{\mathrm{az,el}\},1},\ldots,\hat{\phi}_{\{\mathrm{az,el}\},k}\big)$ and $\hat{\mathtt{v}}_\mathrm{az,el}^k=\big(\hat{\mathtt{v}}_{\{\mathrm{az,el}\},1},\ldots,\hat{\mathtt{v}}_{\{\mathrm{az,el}\},k}\big)$.
\ELSE
\RETURN estimates $\hat{\phi}_\mathrm{az,el}^k=\big\{\big(\frac{2}{\pi},\frac{4}{\pi}\big),\ldots,\big(\frac{2}{\pi},\frac{4}{\pi}\big)\big\}_{k \times 1}$ and $\hat{\mathtt{v}}_\mathrm{az,el}^k=\{(0,0),\ldots,(0,0)\}_{k \times 1}$.
\ENDIF
\end{algorithmic}
\end{algorithm}

The performance of our proposed algorithm (blue squares) can achieve a much finer resolution $\delta$ with the same number of queries as compared to the 5G beam sweeping algorithm. Furthermore, compared with~\cite[Algorithm 1]{zhou2022resolution}, our proposed algorithm has no performance loss, while the computational complexity has decreased, as reflected in terms of runtime of the decoding procedure\footnote{For each number of targets $k$, the two algorithms are run independently $10^3$ times with $d=1$, respectively.} in Table \ref{runtime}. Our proposed algorithm is implemented with brute force random coding. It would be worthwhile to explore more efficient alternatives such as low-density parity-check (LDPC) coding~\cite{gallager1962low} and Polar coding~\cite{arikan2009channel}.

\begin{table}[tb]
\caption{Average Runtime}
\label{runtime}
\centering
\begin{tabular}{|m{4cm}<{\centering}|m{4cm}<{\centering}|m{4cm}<{\centering}|m{4cm}<{\centering}|}
\hline
The Number of Targets & The Number of Queries & Single Threshold & Multiple Thresholds\\
\hline
$k=2$ & $n=60$ & $0.0585~\rms$ & $0.1324~\rms$\\
\hline
$k=3$ & $n=70$ & $0.4145~\rms$ & $1.3518~\rms$\\
\hline
$k=4$ & $n=85$ & $2.2044~\rms$ & $13.9777~\rms$\\
\hline
\end{tabular}
\end{table}

\section{Proof of Non-Asymptotic Bounds (Theorem \ref{th1}-\ref{th3})} \label{sec_proof_nona}

\subsection{Achievability for DMCs (Theorem \ref{th1})} 
\label{Achievable Non-Asymptotic Bound DMC}

We first mathematically describe the non-adaptive query procedure and the excess-resolution events. Given any positive integers $(n,M)\in\bbN^2$, we partition the unit cube $[0,1]^d$ into $(nM)^d$ equal-sized non-overlapping sub-cubes $(\calC_1,\ldots,\calC_{(nM)^d})$. Let there be $k$ targets with initial location and velocity pairs $(\bs,\bv)^k=\{(\bs_1,\bv_1),\ldots,(\bs_k,\bv_k)\}\in([0, 1]^d \times [-v_+,v_+]^d)^k$. Let $\bx=(x^n(1),\ldots,x^n((nM)^d))\in(\{0,1\}^n)^{(nM)^d}$ be the binary query matrix used by the oracle, where at each time point $i\in[n]$, the query is whether or not there exists at least one target located in the region $\calA_i$:
\begin{align}
\calA_i:=\bigcup_{j\in[(nM)^d]:x_{i,j}=1} \calC_j.
\end{align}
The oracle's noiseless response to the query $\calA_i$ is
\begin{align}
z_i:=\bbo\big(\exists~t\in[k]:l(\bs_t,\bv_t,i)\in\calA_i\big) , \label{z}
\end{align}
which is communicated to the questioner as the noisy response $y_i$ at the output of a query-dependent memoryless channel $P_{Y|Z}^{h(|\calA_i|)}$. If there exists a set $\{\hat{\bw}_1^n,\ldots,\hat{\bw}_k^n\}\in\calB_{n,M}^{(k)}$ such that the corresponding mutual information density $\imath^{h(p)}$ satisfies \eqref{single_thres} at single threshold $\gamma$, the decoder $g$ generates $k$ estimated targets $\{(\hat{\bs}_1,\hat{\bv}_1),\ldots,(\hat{\bs}_k,\hat{\bv}_k)\}\in([0, 1]^d \times [-v_+,v_+]^d)^k$ such that $w(\hatbs_t,\hatbv_t,[n])=\hat{\bw}_t^n$ for any $t\in[k]$.

An excess-resolution event occurs if at least one of the following three events occurs:
\begin{itemize}
\item[$\bullet$]$\mathcal{E}_1((\bs,\bv)^k,\mathbf{x},Y^n):$ the true set of $k$ trajectories fails the decoding rule \eqref{single_thres} at time $n$;
\item[$\bullet$]$\mathcal{E}_2((\bs,\bv)^k,\mathbf{x},Y^n):$ there exists a set of $k$ trajectories in which all $k$ trajectories are incorrect passing the decoding rule \eqref{single_thres} at time $n$;
\item[$\bullet$]$\mathcal{E}_3((\bs,\bv)^k,\mathbf{x},Y^n):$ there exists a set of $k$ trajectories in which some trajectories are correct and the other trajectories are incorrect passing the decoding rule \eqref{single_thres} at time $n$.
\end{itemize}

Recall the definition of excess-resolution probability in \eqref{pe}. The excess-resolution probability satisfies the bound:
\begin{align}
\mathrm{P_e}\big((\bs,\bv)^k,\mathbf{x}\big) &:=\mathrm{Pr}\Big\{ \exists~t\in[k]:\min_{(\hatbs,\hatbv)\in(\hatbs,\hatbv)^k} \max_{i\in[0:n]} \big\| l\big(\hatbs,\hatbv,i\big)-l\big(\bs_t,\bv_t,i\big) \big\|_{\infty} > \delta \Big\} \label{pe_event_1}\\
& \leq \mathrm{Pr}\bigg\{\bigcup_{j\in[3]}\mathcal{E}_j\big((\bs,\bv)^k,\mathbf{x},Y^n\big)\bigg\} \label{pe_event_2}\\
& \leq \sum_{j\in[3]} \mathrm{Pr}\big\{\calE_j\big((\bs,\bv)^k,\mathbf{x},Y^n\big)\big\}. \label{pe_event_rlt}
\end{align}

Fix any $p\in(0,1)$. For subsequent analysis, we use the random coding idea by assuming that the query matrix of the oracle is randomly generated i.i.d. from the Bernoulli distribution $P_X=\mathrm{Bern}(p)$. This way, the joint distribution of the query matrix $\bX$ and the noisy response $Y^n$ satisfies that for each $(\bx,y^n)$,
\begin{align}
P_{\bX Y^n}^{\mathrm{md}}(\bx,y^n)=\bigg(\prod_{j\in[(nM)^d]}P_X^n\big(x^n(j)\big)\bigg)\prod_{i\in[n]}P_{Y|Z}^{h(|\calA_i|)}\big(y_i|z_i\big),
\end{align}
where $z_i=\bbo\big(\exists~t\in[k]:x_{i,\Gamma(w(\bs_t,\bv_t,i))}=1\big)$. To apply the change-of-measure technique, we also need the following alternative query-independent distribution
\begin{align}
P_{\mathbf{X}Y^n}^{\mathrm{alt}}(\bx,y^n)=\bigg(\prod_{j\in[(nM)^d]}P_X^n\big(x^n(j)\big)\bigg)\prod_{i\in[n]}P_{Y|Z}^{h(p)}\big(y_i|z_i\big). \label{P_alt}
\end{align}

For any $\eta\in\bbR_+$, define the following typical set of query vectors:
\begin{align}
&\calT^n(M,d,p,\eta):=\Big\{\mathbf{x}\in(\{0,1\}^n)^{(nM)^d}:\max_{i\in[n]}\Big|q_{i,d}^{n,M}(\mathbf{x})-p\Big| \leq \eta\Big\} ,
\end{align}
where 
\begin{align}
q_{i,d}^{n,M}(\mathbf{x}):=\frac{1}{(nM)^d}\sum_{j\in[(nM)^d]} x_{i,j}.
\end{align}
Then, using the random coding idea, analogously to the proof of~\cite[Theorem 1]{zhou2021resolution}, we have
\begin{align}
\mathbb{E}\big[\mathrm{P_e}\big((\bs,\bv)^k,\bX\big)\big] &\leq \mathbb{E}\big[\mathrm{P_e}\big((\bs,\bv)^k,\bX\big)\bbo\big\{\bX\in\calT^n(M,d,p,\eta)\big\}\big]+\mathrm{Pr}\big\{\bX \notin \calT^n(M,d,p,\eta)\big\} \label{pe-second-term-1}\\
&\leq \mathbb{E}\big[\mathrm{P_e}\big((\bs,\bv)^k,\bX\big)\bbo\big\{\bX\in\calT^n(M,d,p,\eta)\big\}\big]+4n\exp\big(-2n^dM^d\eta^2\big) , \label{pe-second-term-2}
\end{align}
where \eqref{pe-second-term-2} follows from~\cite[Lemma 22]{tomamichel2014second}, which upper bounds the second term in \eqref{pe-second-term-1} using the union bound and Hoeffding's inequality.

Consider the fact that $h(\cdot)$ is Lipschitz continuous with parameter $L$ and the assumption on the query-dependent channel in \eqref{dmc}. It follows from the change-of-measure technique that the first term in \eqref{pe-second-term-2} satisfies
\begin{align}
\mathbb{E}\big[\mathrm{P_e}\big((\bs,\bv)^k,\mathbf{X}\big)\bbo\big(\mathbf{X}\in\calT^n(M,d,p,\eta)\big)\big] \leq \exp\big(n \eta Lc(h(p))\big)\sum_{j\in[3]} \underset{P_{\mathbf{X}Y^n}^{\mathrm{alt}}}{\mathrm{Pr}}\big\{\calE_j\big((\bs,\bv)^k,\mathbf{X},Y^n\big)\big\}. \label{typical_pe}
\end{align}

Recall the definition of event $\calG_1(\cdot)$ in \eqref{calG1}. Given any $\gamma\in\mathbb{R}_+$, the probability of the first error event in \eqref{typical_pe} follows that
\begin{align}
\underset{P_{\mathbf{X}Y^n}^{\mathrm{alt}}}{\mathrm{Pr}}\big\{\calE_1\big((\bs,\bv)^k,\mathbf{X},Y^n\big)\big\} 
&=\underset{P_{\mathbf{X}Y^n}^{\mathrm{alt}}}{\mathrm{Pr}}\Big\{ \imath^{h(p)}\big(X^n(\Gamma(\bw_1^n)),\ldots,X^n(\Gamma(\bw_k^n));Y^n\big) < \gamma \Big\}\\*
&=\underset{P_{\mathbf{X}Y^n}^{\mathrm{alt}}}{\mathrm{Pr}}\Big\{ \imath^{h(p)}\big(X^n(1),\ldots,X^n(k);Y^n\big) < \gamma \Big\}\\*
&=\underset{P_{\mathbf{X}Y^n}^{\mathrm{alt}}}{\mathrm{Pr}}\big\{ \calG_1(n,\gamma) \big\}. \label{p-event-1-final}
\end{align}

Recall the definition of the set of given $k$ targets’ quantized trajectories in \eqref{set-u}. For subsequent derivation of the probability of the second and third error events in \eqref{typical_pe}, define the following sets when $\calJ \subset [k]$,
\begin{align}
&\calI\big((\bs,\bv)^k,\calJ\big):=\Big\{ \big\{\bw_1^n,\ldots,\bw_k^n\big\}\in\calB_{n,M}^{(k)}:\bw_t^n\in\calU_{n,M}\big((\bs,\bv)^k\big)~\mathrm{for}~t\in\calJ,\bw_t^n \notin \calU_{n,M}\big((\bs,\bv)^k\big)~\mathrm{for}~t\in([k]\setminus\calJ) \Big\}. \label{ij}
\end{align}
Note that $\calI((\bs,\bv)^k,\calJ)$ denotes such a set of trajectories that some elements are same as the elements of true targets' trajectories set with the index set $\calJ$ and others are different from the rest elements of true targets' trajectories set. When $\calJ=\emptyset$, $\calI((\bs,\bv)^k,\emptyset)$ denotes the set of $k$-length sets with $k$ different elements from the true targets' trajectories set $\calU_{n,M}((\bs,\bv)^k)$. The size of $\calI((\bs,\bv)^k,\emptyset)$ satisfies
\begin{align}
\big|\calI((\bs,\bv)^k,\emptyset)\big| = \binom{\big((2nv_++3)n^4M^2\big)^d-k}{k} , \label{size-I}
\end{align}
where \eqref{size-I} follows since the size of $\calI((\bs,\bv)^k,\emptyset)$ is equal to the number of all possible combinations of choosing $k$ trajectories from $((2nv_++3)n^4M^2)^d-k$ non-target trajectories.

For any $\widetilde{\bw}_{[k]}^n=\{\widetilde{\bw}_1^n,\ldots,\widetilde{\bw}_k^n\}\in\calI((\bs,\bv)^k,\emptyset)$, we denote $(X^n(\Gamma(\widetilde{\bw}_1^n)),\ldots,X^n(\Gamma(\widetilde{\bw}_k^n)))$ with $X_{\widetilde{\bw}_{[k]}^n}^n$. Recall the definition of function $\bbH_1(\cdot)$ in \eqref{calH1}. We can upper bound $\mathrm{Pr}\{\calE_2((\bs,\bv)^k,\mathbf{X},Y^n)\}$ by using the information spectrum method introduced in~\cite{han2006information} (see also~\cite[Lemma 7.10.1]{koga2002information}) as follows:
\begin{align}
&\underset{P_{\mathbf{X}Y^n}^{\mathrm{alt}}}{\mathrm{Pr}}\big\{\calE_2\big((\bs,\bv)^k,\mathbf{X},Y^n\big)\big\} \nn\\
&= \underset{P_{\mathbf{X}Y^n}^{\mathrm{alt}}}{\mathrm{Pr}}\Big\{\exists~\widetilde{\bw}_{[k]}^n\in\calI\big((\bs,\bv)^k,\emptyset\big): \imath^{h(p)}\big(X_{\widetilde{\bw}_{[k]}^n}^n;Y^n\big) \geq \gamma\Big\} \label{p-event-2-1}\\
&\leq \sum_{\widetilde{\bw}_{[k]}^n\in\calI((\bs,\bv)^k,\emptyset)}\underset{P_{\mathbf{X}Y^n}^{\mathrm{alt}}}{\mathrm{Pr}}\Big\{\imath^{h(p)}\big(X_{\widetilde{\bw}_{[k]}^n}^n;Y^n\big) \geq \gamma\Big\} \label{p-event-2-2}\\
&= \sum_{\widetilde{\bw}_{[k]}^n\in\calI((\bs,\bv)^k,\emptyset)} \sum_{(\mathbf{x},y^n):\imath^{h(p)}(x_{\widetilde{\bw}_{[k]}^n}^n;y^n) \geq \gamma} P_{\mathbf{X}Y^n}^{\mathrm{alt}}(\mathbf{x},y^n) \label{p-event-2-3}\\
&= \sum_{\widetilde{\bw}_{[k]}^n\in\calI((\bs,\bv)^k,\emptyset)} \sum_{\imath^{h(p)}(x_{\widetilde{\bw}_{[k]}^n}^n;y^n) \geq \gamma} \bigg(\prod_{t\in[k]}P_X^n\big(x^n\big(\Gamma(\widetilde{\bw}_t^n)\big)\big)\bigg) \times P_{Y^n}^{\mathrm{alt}}(y^n) \label{p-event-2-4}\\
&\leq \sum_{\widetilde{\bw}_{[k]}^n\in\calI((\bs,\bv)^k,\emptyset)} \sum_{\imath^{h(p)}(x_{\widetilde{\bw}_{[k]}^n}^n;y^n) \geq \gamma} \bigg(\prod_{t\in[k]}P_X^n\big(x^n\big(\Gamma(\widetilde{\bw}_t^n)\big)\big)\bigg) \times P_{Y^n|X_{\widetilde{\bw}_{[k]}^n}^n}^{\mathrm{alt}}\big(y^n|x_{\widetilde{\bw}_{[k]}^n}^n\big) \times \exp(-\gamma) \label{p-event-2-5}\\
&= \sum_{\widetilde{\bw}_{[k]}^n\in\calI((\bs,\bv)^k,\emptyset)} \exp(-\gamma) \label{p-event-2-6}\\
&= \binom{\big((2nv_++3)n^4M^2\big)^d-k}{k} \exp(-\gamma) \label{p-event-2-7}\\
&=\bbH_1\Big(\big((2nv_++3)n^4M^2\big)^d-k,k,\gamma\Big) , \label{p-event-2-final}
\end{align}
where \eqref{p-event-2-1} follows from the definition of $\calI((\bs,\bv)^k,\calJ)$ in \eqref{ij}, \eqref{p-event-2-4} follows from the fact that the noisy $Y^n$ is only dependent on binary vectors $X_{\calU_{n,M}((\bs,\bv)^k)}^n$, \eqref{p-event-2-5} follows from the Algorithm \ref{alg:1} in Section \ref{mr Non-Asymptotic Bounds}, \eqref{p-event-2-7} follows from the size of the set $\calI((\bs,\bv)^k,\emptyset)$ in \eqref{size-I}.

Then, we focus on the probability of the occurrence of the third error event. Let $\calJ\in\Lambda([k])$ be the set of estimated target indices which are same as the true targets and let $\calJ^c$ denote $[k]\setminus\calJ$. Define
\begin{align}
\widetilde{\calI}\big((\bs,\bv)^k,\calJ^c\big):=\Big\{ \big\{\bw_1^n,\ldots,\bw_{|\calJ^c|}^n\big\}\in\calB_{n,M}^{(|\calJ^c|)}:\bw_t^n \notin \calU_{n,M}\big((\bs,\bv)^k\big)~\mathrm{for}~t\in\calJ^c \Big\}. \label{w-calJ}
\end{align}
Its size satisfies
\begin{align}
\big|\widetilde{\calI}\big((\bs,\bv)^k,\calJ^c\big)\big| = \binom{\big((2nv_++3)n^4M^2\big)^d-k}{|\calJ^c|} , \label{size-I-tilde}
\end{align}
where \eqref{size-I-tilde} follows since the size of $\widetilde{\calI}((\bs,\bv)^k,\calJ^c)$ is equal to the number of all possible combinations of choosing $|\calJ^c|$ trajectories from $((2nv_++3)n^4M^2)^d-k$ non-target trajectories.

Next, set an arbitrary $\lambda_\calJ > 0$, define the following two events
\begin{align}
&\widetilde{\calE}_1(\calJ):=\Big\{ \imath^{h(p)}\big(X_{\calJ}^n;Y^n\big) > nC(p,|\calJ|)+n\lambda_\calJ \Big\} , \label{w-event1}\\
&\widetilde{\calE}_2(\calJ):=\bigcup_{\bw_{\calJ^c}^n\in\widetilde{\calI}((\bs,\bv)^k,\calJ^c)} \Big\{ \imath^{h(p)}\big(\bar{X}_{\calJ^c}^n(\Gamma(\bw_{\calJ^c}^n)),X_{\calJ}^n;Y^n\big) \geq \gamma \Big\} , \label{w-event2}
\end{align}
and the threshold $\gamma_\calJ$ corresponding to the set $\calJ$
\begin{align}
\gamma_\calJ:=\gamma-nC(p,|\calJ|)-n\lambda_\calJ. \label{gammaJ}
\end{align}

Therefore, we can upper bound $\mathrm{Pr}\{\calE_3((\bs,\bv)^k,\mathbf{X},Y^n)\}$ as follows:
\begin{align}
\underset{P_{\mathbf{X}Y^n}^{\mathrm{alt}}}{\mathrm{Pr}}\big\{\calE_3\big((\bs,\bv)^k,\mathbf{X},Y^n\big)\big\} &\leq \sum_{\calJ\in\Lambda([k])} \underset{P_{\mathbf{X}Y^n}^{\mathrm{alt}}}{\mathrm{Pr}}\big\{\widetilde{\calE}_2(\calJ)\big\} \label{p-event-3-1}\\
&= \sum_{\calJ\in\Lambda([k])} \bigg(\underset{P_{\mathbf{X}Y^n}^{\mathrm{alt}}}{\mathrm{Pr}}\big\{\widetilde{\calE}_1(\calJ)\cap\widetilde{\calE}_2(\calJ)\big\}+\underset{P_{\mathbf{X}Y^n}^{\mathrm{alt}}}{\mathrm{Pr}}\big\{\widetilde{\calE}_1^c(\calJ)\cap\widetilde{\calE}_2(\calJ)\big\}\bigg). \label{p-event-3-2}
\end{align}
Recall the definition of event $\calG_2(\cdot)$ in \eqref{calG2}. The first term in the brace in \eqref{p-event-3-2} can be further upper bounded as follows:
\begin{align}
\underset{P_{\mathbf{X}Y^n}^{\mathrm{alt}}}{\mathrm{Pr}}\big\{\widetilde{\calE}_1(\calJ)\cap\widetilde{\calE}_2(\calJ)\big\} &\leq \underset{P_{\mathbf{X}Y^n}^{\mathrm{alt}}}{\mathrm{Pr}}\big\{\widetilde{\calE}_1(\calJ)\big\} \\
&\leq \underset{P_{\mathbf{X}Y^n}^{\mathrm{alt}}}{\mathrm{Pr}}\Big\{\imath^{h(p)}\big(X_{\calJ}^n;Y^n\big) > nC(p,|\calJ|)+n\lambda_\calJ\Big\} \\
&= \underset{P_{\mathbf{X}Y^n}^{\mathrm{alt}}}{\mathrm{Pr}}\big\{\calG_2\big(n,\calJ,\lambda_{\calJ}\big)\big\},\label{p-event-3-left-rlt}
\end{align}
where \eqref{p-event-3-left-rlt} follows from the definition of event $\widetilde{\calE}_1(\calJ)$ in \eqref{w-event1}. The second term in the brace in \eqref{p-event-3-2} can be further upper bounded as follows:
\begin{align}
&\underset{P_{\mathbf{X}Y^n}^{\mathrm{alt}}}{\mathrm{Pr}}\big\{\widetilde{\calE}_1^c(\calJ)\cap\widetilde{\calE}_2(\calJ)\big\} \nn\\
&= \underset{P_{\mathbf{X}Y^n}^{\mathrm{alt}}}{\mathrm{Pr}}\Bigg\{ \Big\{\imath^{h(p)}\big(X_{\calJ}^n;Y^n\big) \leq nC(p,|\calJ|)+n\lambda_\calJ\Big\}\bigcap\Bigg\{\bigcup_{\bw_{\calJ^c}^n\in\widetilde{\calI}((\bs,\bv)^k,\calJ^c)} \Big\{ \imath^{h(p)}\big(\bar{X}_{\calJ^c}^n(\Gamma(\bw_{\calJ^c}^n)),X_{\calJ}^n;Y^n\big) \geq \gamma \Big\}\Bigg\} \Bigg\} \label{p-event-3-right-1}\\
&\leq \underset{P_{\mathbf{X}Y^n}^{\mathrm{alt}}}{\mathrm{Pr}}\Bigg\{\bigcup_{\bw_{\calJ^c}^n\in\widetilde{\calI}((\bs,\bv)^k,\calJ^c)} \Big\{ \imath_{\calJ}^{h(p)}\big(\bar{X}_{\calJ^c}^n(\Gamma(\bw_{\calJ^c}^n)),X_{\calJ}^n;Y^n\big) \geq \gamma_\calJ \Big\}\Bigg\} \label{p-event-3-right-2}\\
&\leq \binom{\big((2nv_++3)n^4M^2\big)^d-k}{|\calJ^c|} \underset{P_{\mathbf{X}Y^n}^{\mathrm{alt}}}{\mathrm{Pr}}\Big\{ \imath_{\calJ}^{h(p),k}\big(\bar{X}_{\calJ^c}^n,X_{\calJ}^n;Y^n\big) \geq \gamma_\calJ \Big\} \label{p-event-3-right-3}\\
&\leq \binom{\big((2nv_++3)n^4M^2\big)^d-k}{k-|\calJ|}\exp(-\gamma_\calJ) \label{p-event-3-right-4}\\
&\leq \binom{\big((2nv_++3)n^4M^2\big)^d-k}{k-|\calJ|}\exp(-\gamma+nC(p,|\calJ|)+n\lambda_\calJ) \label{p-event-3-right-5}\\
&= \bbH_1\Big(\big((2nv_++3)n^4M^2\big)^d-k,k-|\calJ|,\gamma-nC(p,|\calJ|)-n\lambda_\calJ\Big) , \label{p-event-3-right-rlt}
\end{align}
where \eqref{p-event-3-right-1} follows from the two events $\widetilde{\calE}_1(\calJ),\widetilde{\calE}_2(\calJ)$ defined in \eqref{w-event1}, \eqref{w-event2}, \eqref{p-event-3-right-3} follows from the size of the set $\widetilde{\calI}((\bs,\bv)^k,\calJ^c)$ in \eqref{size-I-tilde}, \eqref{p-event-3-right-4} follows from~\cite[Eq. 88]{vincent2018on}, \eqref{p-event-3-right-5} follows from the definition of $\gamma_\calJ$ in \eqref{gammaJ}, \eqref{p-event-3-right-2} follows since
\begin{align}
\imath_{\calJ}^{h(p)}\big(\bar{X}_{\calJ^c}^n(\Gamma(\bw_{\calJ^c}^n)),X_{\calJ}^n;Y^n\big) &:= \log \frac{P_{Y | X_{[k]}}^{h(p)}\big(y | x_{[k]}\big)}{P_{Y | X_{\calJ}}^{h(p)}\big(y | x_{\calJ}\big)} \\*
&= \log \frac{P_{Y | X_{[k]}}^{h(p)}\big(y | x_{[k]}\big)}{P_Y^{h(p)}(y)} - \log \frac{P_{Y | X_{\calJ}}^{h(p)}\big(y | x_{\calJ}\big)}{P_Y^{h(p)}(y)} \\*
&= \imath^{h(p)}\big(\bar{X}_{\calJ^c}^n(\Gamma(\bw_{\calJ^c}^n)),X_{\calJ}^n;Y^n\big) - \imath^{h(p)}\big(X_{\calJ}^n;Y^n\big) \\*
&\geq \gamma - nC(p,|\calJ|) - n\lambda_\calJ \label{gamma-nc-nlambda}\\*
&= \gamma_\calJ,
\end{align}
where \eqref{gamma-nc-nlambda} follows from the two inequalities in \eqref{p-event-3-right-1}, which state that $\imath^{h(p)}\big(X_{\calJ}^n;Y^n\big) \leq nC(p,|\calJ|)+n\lambda_\calJ$ and\\$\imath^{h(p)}\big(\bar{X}_{\calJ^c}^n(\Gamma(\bw_{\calJ^c}^n)),X_{\calJ}^n;Y^n\big) \geq \gamma$ .

Combining \eqref{pe-second-term-2}, \eqref{typical_pe}, \eqref{p-event-1-final}, \eqref{p-event-2-final}, \eqref{p-event-3-2}, \eqref{p-event-3-left-rlt} and \eqref{p-event-3-right-rlt}, we conclude that for any $k$ targets' initial location and velocity pairs $(\bS,\bV)^k$, the excess-resolution probability in Algorithm \ref{alg:1} satisfies
\begin{align}
&\mathbb{E}\Big[\mathrm{P_e}\Big(\big(\bS,\bV\big)^k,\mathbf{X}\Big)\Big] \nn\\*
&\leq 4n\exp\big(-2n^dM^d\eta^2\big) + \exp\big(n \eta Lc(h(p))\big) \times \bigg( \underset{(P_{X_{[k]} Y}^{h(p)})^n}{\mathrm{Pr}}\big\{ \calG_1(n,\gamma) \big\} + \bbH_1\Big(\big((2nv_++3)n^4M^2\big)^d-k,k,\gamma\Big) \bigg. \nn\\*
&\qquad \bigg.+ \sum_{\calJ\in\Lambda([k])} \bigg(\underset{(P_{X_{[k]} Y}^{h(p)})^n}{\mathrm{Pr}}\big\{ \calG_2\big(n,\calJ,\lambda_{\calJ}\big) \big\} + \bbH_1\Big(\big((2nv_++3)n^4M^2\big)^d-k,k-|\calJ|,\gamma-nC(p,|\calJ|)-n\lambda_\calJ\Big) \bigg) \bigg) , \label{epe-con-1}
\end{align}
where \eqref{epe-con-1} follows from the definition of the joint distributions of $P_{X_{[k]} Y}^{h(p)}$ in \eqref{joint_d} and $P_{\mathbf{X}Y^n}^{\mathrm{alt}}$ in \eqref{P_alt}. 

\subsection{Achievability for AWGN Channels (Theorem \ref{th2})} 
\label{Achievable Non-Asymptotic Bound AWGN}

In this subsection, we will focus on explaining the changes in the achievable non-asymptotic bound derivation for a query-dependent AWGN channel.

Fix $p\in(0,1)$. Let $\psi^n$ be generated from the Gaussian distribution $\calN(0,\sigma^2)$ with mean $0$ and variance $\sigma^2$. Define $\zeta$ as 
\begin{align}
\zeta:=2\sigma^2\big(h(p)^2+L\eta(2h(p)+L\eta)\big). \label{alpha}
\end{align}
Recall the definition of $z$ in \eqref{z}. For any $(\bs,\bv)^k\in([0,1]^d \times [-v_+,v_+]^d)^k$ and any $\bx\in\calT^n(M,d,p,\eta)$, given queries $\calA^n=(\calA_1,\ldots,\calA_n)$, we have that
\begin{align}
\mathrm{Pr}\big\{\big\|Y^n-z^n\big\|^2>n\zeta\big\} &\leq \exp\Big(-\frac{n(1-\log2)}{2}\Big) , \label{Pryminusz-0}
\end{align}
where \eqref{Pryminusz-0} follows from~\cite[Section IV-B]{zhou2023resolution}.

The analysis for the query-dependent AWGN channel case is exactly the same as that for the query-dependent DMC case provided in Section \ref{Achievable Non-Asymptotic Bound DMC} till \eqref{pe-second-term-2}. Then we need to bound the first term in \eqref{pe-second-term-2} in a slightly different manner as follows:
\begin{align}
&\mathbb{E}\big[\mathrm{P_e}\big((\bs,\bv)^k,\mathbf{X}\big)\bbo\big(\mathbf{X}\in\calT^n(M,d,p,\eta)\big)\big] \nn\\
&= \mathbb{E}\big[\mathrm{P_e}\big((\bs,\bv)^k,\mathbf{X}\big)\bbo\big(\mathbf{X}\in\calT^n(M,d,p,\eta)\big)\bbo\big(\big\|Y^n-Z^n\big\|^2>n\zeta\big)\big] \nn\\
&+ \mathbb{E}\big[\mathrm{P_e}\big((\bs,\bv)^k,\mathbf{X}\big)\bbo\big(\mathbf{X}\in\calT^n(M,d,p,\eta)\big)\bbo\big(\big\|Y^n-Z^n\big\|^2 \leq n\zeta\big)\big] \label{typical-pe-AWGN-1}\\
&\leq \mathbb{E}\big[\bbo\big(\mathbf{X}\in\calT^n(M,d,p,\eta)\big)\bbo\big(\big\|Y^n-Z^n\big\|^2>n\zeta\big)\big] + \mathbb{E}\big[\mathrm{P_e}\big((\bs,\bv)^k,\mathbf{X}\big)\bbo\big(\mathbf{X}\in\calT^n(M,d,p,\eta)\big)\bbo\big(\big\|Y^n-Z^n\big\|^2 \leq n\zeta\big)\big] \label{typical-pe-AWGN-2}\\
&\leq \mathbb{E}\big[\mathrm{P_e}\big((\bs,\bv)^k,\mathbf{X}\big)\bbo\big(\mathbf{X}\in\calT^n(M,d,p,\eta)\big)\bbo\big(\big\|Y^n-Z^n\big\|^2 \leq n\zeta\big)\big] + \exp\Big(-\frac{n(1-\log 2)}{2}\Big) , \label{typical-pe-AWGN-0}
\end{align}
where \eqref{typical-pe-AWGN-0} follows from \eqref{Pryminusz-0}.

For query-dependent AWGN channel, when $\mathbf{X}\in\calT(M,d,p,\eta)$ and $\|Y^n-Z^n\|^2 \leq n\zeta$, analogously to~\cite[Section IV-B]{zhou2023resolution}, we have
\begin{align}
\log \frac{P_{\mathbf{X}Y^n}^{\mathrm{md}}(\mathbf{x},y^n)}{P_{\mathbf{X}Y^n}^{\mathrm{alt}}(\mathbf{x},y^n)} &= \sum_{i\in[n]} \log \frac{h(p)}{h(|\calA_i|)} + \sum_{i\in[n]} (y_i-z_i)^2\frac{h(|\calA_i|)^2-h(p)^2}{2\sigma^2h(p)^2h(|\calA_i|)^2} \\
&\leq \frac{nL\eta}{h(p)-L\eta} + \frac{nL\eta\big(h(p)^2+L\eta(2h(p)+L\eta)\big)(2h(p)+L\eta)}{h(p)^2\big(h(p)^2-L\eta(2h(p)+L\eta)\big)} = \Xi(n,p,\sigma,\eta). \label{aleph-rlt}
\end{align}

Then using the change-of-measure technique derived in \eqref{aleph-rlt}, the first term in \eqref{typical-pe-AWGN-0} is further upper bounded as follows:
\begin{align}
\mathbb{E}\big[\mathrm{P_e}\big((\bs,\bv)^k,\mathbf{X}\big)\bbo\big(\mathbf{X}\in\calT^n(M,d,p,\eta)\big)\bbo\big(\big\|Y^n-Z^n\big\|^2 \leq n\zeta\big)\big] \leq \exp\big(\Xi(n,p,\sigma,\eta)\big)\sum_{j\in[3]} \underset{P_{\mathbf{X}Y^n}^{\mathrm{alt}}}{\mathrm{Pr}}\big\{\calE_j\big((\bs,\bv)^k,\mathbf{X},Y^n\big)\big\}.
\end{align}
The subsequent analysis is consistent with the analysis of the probability of three excess-resolution events in Section \ref{Achievable Non-Asymptotic Bound DMC}.

\subsection{Converse} \label{Non-Asymptotic Converse Bound}

We now prove the non-asymptotic converse result that bounds the performance of an optimal non-adaptive query procedure.

\subsubsection{From Estimation of Trajectory to Estimation Initial Locations and Velocities}

In this subsection, we show that the excess-resolution probability of estimating the trajectories can be further lower bounded by the excess-resolution probability of estimating the initial locations and velocities. To do so, we need to define the following events
\begin{align}
&\calE_\dagger:=\bigg\{ \big(\hatbS,\hatbV\big)^k:\exists~t\in[k]:\min_{(\hatbs,\hatbv)\in(\hatbS,\hatbV)^k} \big\|\hatbs-\bS_t\big\|_{\infty} > \delta \bigg\} , \\
&\calE_\ddagger:=\bigg\{ \big(\hatbS,\hatbV\big)^k:\exists~t\in[k]:\min_{(\hatbs,\hatbv)\in(\hatbS,\hatbV)^k, \|\hatbs-\bS_t\|_{\infty} \leq \delta} n\big\|\hatbv-\bV_t\big\|_{\infty} > 2\delta \bigg\}.
\end{align}

Note that $\calE_\dagger$ leads to an excess-resolution event at time $i=0$. Using the triangle inequality, it follows that if $(\hatbs,\hatbv)^k\in\calE_\ddagger$, $\exists~t\in[k]$,
\begin{align}
\min_{(\hatbs,\hatbv)\in(\hatbs,\hatbv)^k} \big\| l(\hatbs,\hatbv,n)-l\big(\bS_t,\bV_t,n\big) \big\|_{\infty} &= \min_{(\hatbs,\hatbv)\in(\hatbs,\hatbv)^k} \max_{d^{\prime}\in[d]} \big| (\hat{s}_{d^{\prime}}+n\hat{v}_{d^{\prime}}) - (S_{t,d^{\prime}}+nV_{t,d^{\prime}}) \big| \\
&\geq \min_{(\hatbs,\hatbv)\in(\hatbs,\hatbv)^k} \max_{d^{\prime}\in[d]} \big(n\big|\hat{v}_{d^{\prime}}-V_{t,d^{\prime}}\big|-\big|\hat{s}_{d^{\prime}}-S_{t,d^{\prime}}\big|\big) \\
&> 2\delta-\delta \\
&= \delta ,
\end{align}
which leads to an excess-resolution event at time $i=n$.

Combining the above arguments and using \eqref{pe}, it follows that
\begin{align}
\varepsilon &\geq \mathrm{Pr}\bigg\{\exists~t\in[k]:\min_{(\hatbs,\hatbv)\in(\hatbs,\hatbv)^k} \max_{i\in[0:n]} \big\| l\big(\hatbs,\hatbv,i\big)-l\big(\bS_t,\bV_t,i\big) \big\|_{\infty} > \delta\bigg\} \\
&\geq \mathrm{Pr}\Big\{\big(\hatbs,\hatbv\big)^k\in\big(\calE_\dagger \cup \calE_\ddagger\big)\Big\} \\
&= \mathrm{Pr}\bigg\{\exists~t\in[k]:\min_{(\hatbs,\hatbv)\in(\hatbs,\hatbv)^k} \big\| \hatbs-\bS_t \big\|_{\infty} > \delta~\mathrm{or}~\min_{(\hatbs,\hatbv)\in(\hatbs,\hatbv)^k} \big\| \hatbv-\bV_t \big\|_{\infty} > \frac{2\delta}{n}\bigg\}. \label{pe-epsilon}
\end{align}

\subsubsection{Connection to Channel Coding}\

Let $\beta\in\bbR_+$ be arbitrary such that $\beta < \frac{1-\varepsilon}{2} < 0.25$ and let $\tilde{M}:=\big\lfloor\frac{\beta}{\delta}\big\rfloor$. Partition the set $[0,1]$ into $\tilde{M}$ disjoint equal size subsets $\calS_1,\ldots,\calS_{\tilde{M}}$. Furthermore, partition the set $[-v_+,v_+]$ into $\tilde{M}_* = \lfloor2nv_+\tilde{M}\rfloor$ disjoint equal size subset $\calV_1,\ldots,\calV_{\tilde{M}_*}$. Now define the following quantization functions
\begin{align}
\mathrm{q_s}(s) &:= \sum_{j\in[\tilde{M}]} j\bbo(s\in\calS_j), \forall~s\in[0,1] , \\
\mathrm{q_v}(v) &:= \sum_{j\in[\tilde{M}_*]} j\bbo(v\in\calV_j), \forall~v\in[-v_+,v_+].
\end{align}
In subsequent analysis, assuming that the initial locations and velocities of $k$ targets are distinct, let
\begin{align}
\Omega_0^k\big(\bS^k\big) &:= \big\{\big(\mathrm{q_s}(S_{1,1}),\ldots,\mathrm{q_s}(S_{1,d})\big),\ldots,\big(\mathrm{q_s}(S_{k,1}),\ldots,\mathrm{q_s}(S_{k,d})\big)\big\}=\big\{\Omega_{0,1},\ldots,\Omega_{0,k}\big\} ,\\
\Omega_1^k\big(\bV^k\big) &:= \big\{\big(\mathrm{q_v}(V_{1,1}),\ldots,\mathrm{q_v}(V_{1,d})\big),\ldots,\big(\mathrm{q_v}(V_{k,1}),\ldots,\mathrm{q_v}(V_{k,d})\big)\big\}=\big\{\Omega_{1,1},\ldots,\Omega_{1,k}\big\} ,\\
\hat{\Omega}_0^k\big(\hatbS^k\big) &:= \Big\{\Big(\mathrm{q_s}\big(\hat{S}_{1,1}\big),\ldots,\mathrm{q_s}\big(\hat{S}_{1,d}\big)\Big),\ldots,\Big(\mathrm{q_s}\big(\hat{S}_{k,1}\big),\ldots,\mathrm{q_s}\big(\hat{S}_{k,d}\big)\Big)\Big\}=\big\{\hat{\Omega}_{0,1},\ldots,\hat{\Omega}_{0,k}\big\} ,\\
\hat{\Omega}_1^k\big(\hatbV^k\big) &:= \Big\{\Big(\mathrm{q_v}\big(\hat{V}_{1,1}\big),\ldots,\mathrm{q_v}\big(\hat{V}_{1,d}\big)\Big),\ldots,\Big(\mathrm{q_v}\big(\hat{V}_{k,1}\big),\ldots,\mathrm{q_v}\big(\hat{V}_{k,d}\big)\Big)\Big\}=\big\{\hat{\Omega}_{1,1},\ldots,\hat{\Omega}_{1,k}\big\} .
\end{align}

Given queries $\calA^n=(\calA_1,\ldots,\calA_n)$, the noiseless answer $z_i$ to query $\calA_i$ is
\begin{align}
z_i:=\bbo\big(\exists~t\in[k]:l(\bS_t,\bV_t,i)\in\calA_i\big).
\end{align}
When we consider the random pairs of targets' initial locations and velocities $(\bS,\bV)^k$, the induced random noiseless response $Z_i\sim\mathrm{Bern}(1-(1-|\calA_i|)^k)$. This is because $Z_i=1$ if for any $t\in[k]$, the real-time location $l(\bS_t,\bV_t,i)$ lies in the region $\calA_i$. Equivalently, $Z_i$ is the binary OR of $k$ independent random variables $\big(X_{i,\Gamma(w(\bS_1,\bV_1,i))},\ldots,X_{i,\Gamma(w(\bS_k,\bV_k,i))}\big)$, each generated from Bern$(|\calA_i|)$, i.e., $Z_i=\bbo\big(\exists~t\in[k]:X_{i,\Gamma(w(\bS_t,\bV_t,i))}=1\big)$. The noisy response $Y_i$ is the output of passing $Z_i$ over the query-dependent channel $P_{Y|Z}^{h(|\calA_i|)}$. The marginal distribution $P_{Y_i}^{\calA_i}$ is thus induced by the distribution of $Z_i$ and the channel $P_{Y|Z}^{h(|\calA_i|)}$.

From the problem formulation, the Markov chain $(\Omega_0^k,\Omega_1^k)-(\bS,\bV)^k-Z^n-Y^n-(\hatbS,\hatbV)^k$ holds and the joint distribution of these random variables satisfies
\begin{align}
&\mathrm{P}_{\Omega_0^k\Omega_1^k(\bS\bV)^kZ^nY^n(\hatbS\hatbV)^k}\big(\omega_0^k,\omega_1^k,(\bs,\bv)^k,z^n,y^n,(\hatbs,\hatbv)^k\big) \nn\\*
&= f_{\bS^k\bV^k}\big((\bs,\bv)^k\big)\prod_{t\in[k]}\Big(\bbo\big(\omega_{0,t}=\mathrm{q_s}(\bs_t)\big) \times \bbo\big(\omega_{1,t}=\mathrm{q_v}(\bv_t)\big)\Big) \nn\\*
&\qquad\qquad\qquad\qquad\qquad\times \Bigg( \prod_{i\in[n]} \bbo\big(z_i=\bbo\big(\exists~t\in[k],l(\bs_t,\bv_t,i)\in\calA_i\big)\big)P_{Y|Z}^{\calA_i}(y_i|z_i) \Bigg) \bbo\big((\hatbs,\hatbv)^k=g(y^n)\big). \label{markov}
\end{align}
Unless stated otherwise, the probabilities of events are calculated according to the distribution in \eqref{markov} or its induced marginal and conditional versions.

Let $\hat{\Omega}_{[0:1]}=(\hat{\Omega}_0^k,\hat{\Omega}_1^k)$ and $\Omega_{[0:1]}=(\Omega_0^k,\Omega_1^k)$, it follows that
\begin{align}
&\mathrm{Pr}\Big\{\Omega_{[0:1]} \neq \hat{\Omega}_{[0:1]}\Big\} \nn\\
&=\mathrm{Pr}\bigg\{ \Omega_{[0:1]} \neq \hat{\Omega}_{[0:1]}~\mathrm{and}~\exists~t\in[k]: \min_{(\hatbs,\hatbv)\in(\hatbs,\hatbv)^k} \big\|\hatbs-\bS_t\big\|_{\infty} > \delta~\mathrm{or}~\min_{(\hatbs,\hatbv)\in(\hatbs,\hatbv)^k} \big\|\hatbv-\bV_t\big\|_{\infty} > \frac{2\delta}{n} \bigg\} \nn\\
&\qquad\qquad+\mathrm{Pr}\bigg\{ \Omega_{[0:1]} \neq \hat{\Omega}_{[0:1]}~\mathrm{and}~\forall~t\in[k]: \min_{(\hatbs,\hatbv)\in(\hatbs,\hatbv)^k} \big\|\hatbs-\bS_t\big\|_{\infty} \leq \delta~\mathrm{and}~\min_{(\hatbs,\hatbv)\in(\hatbs,\hatbv)^k} \big\|\hatbv-\bV_t\big\|_{\infty} \leq \frac{2\delta}{n} \bigg\} \label{Omega-1}\\
&\leq \mathrm{Pr}\bigg\{ \exists~t\in[k]: \min_{(\hatbs,\hatbv)\in(\hatbs,\hatbv)^k} \big\|\hatbs-\bS_t\big\|_{\infty} > \delta~\mathrm{or}~\min_{(\hatbs,\hatbv)\in(\hatbs,\hatbv)^k} \big\|\hatbv-\bV_t\big\|_{\infty} > \frac{2\delta}{n} \bigg\} \nn\\
&\qquad\qquad+\mathrm{Pr}\bigg\{ \Omega_{[0:1]} \neq \hat{\Omega}_{[0:1]}~\mathrm{and}~\forall~t\in[k]: \min_{(\hatbs,\hatbv)\in(\hatbs,\hatbv)^k} \big\|\hatbs-\bS_t\big\|_{\infty} \leq \delta~\mathrm{and}~\min_{(\hatbs,\hatbv)\in(\hatbs,\hatbv)^k} \big\|\hatbv-\bV_t\big\|_{\infty} \leq \frac{2\delta}{n} \bigg\} \label{Omega-2}\\
&\leq \varepsilon+\mathrm{Pr}\bigg\{ \Omega_{[0:1]} \neq \hat{\Omega}_{[0:1]}~\mathrm{and}~\forall~t\in[k]: \min_{(\hatbs,\hatbv)\in(\hatbs,\hatbv)^k} \big\|\hatbs-\bS_t\big\|_{\infty} \leq \delta~\mathrm{and}~\min_{(\hatbs,\hatbv)\in(\hatbs,\hatbv)^k} \big\|\hatbv-\bV_t\big\|_{\infty} \leq \frac{2\delta}{n} \bigg\} \label{Omega-3}\\
&\leq \varepsilon+\mathrm{Pr}\bigg\{ \Omega_0 \neq \hat{\Omega}_0~\mathrm{and}~\forall~t\in[k]: \min_{(\hatbs,\hatbv)\in(\hatbs,\hatbv)^k} \big\|\hatbs-\bS_t\big\|_{\infty} \leq \delta \bigg\} \nn\\
&\qquad\qquad+\mathrm{Pr}\bigg\{ \Omega_1 \neq \hat{\Omega}_1~\mathrm{and}~\forall~t\in[k]:\min_{(\hatbs,\hatbv)\in(\hatbs,\hatbv)^k} \big\|\hatbv-\bV_t\big\|_{\infty} \leq \frac{2\delta}{n} \bigg\} \label{Omega-4}\\
&\leq \varepsilon+\sum_{t\in[k]}\bigg(\mathrm{Pr}\bigg\{ \omega_{0,t}\in\hat{\Omega}_0~\mathrm{and}~\min_{(\hatbs,\hatbv)\in(\hatbs,\hatbv)^k} \big\|\hatbs-\bS_t\big\|_{\infty} \leq \delta \bigg\}+\mathrm{Pr}\bigg\{\omega_{1,t}\in\hat{\Omega}_1~\mathrm{and}~\min_{(\hatbs,\hatbv)\in(\hatbs,\hatbv)^k} \big\|\hatbv-\bV_t\big\|_{\infty} \leq \frac{2\delta}{n} \bigg\}\bigg) \label{Omega-5}\\
&=\varepsilon+\sum_{t\in[k]}\bigg(\mathrm{Pr}\bigg\{ \omega_{0,t}\in\hat{\Omega}_0~\mathrm{and}~\exists~(\hatbs,\hatbv)\in\big(\hatbs,\hatbv\big)^k:\big\|\hatbs-\bS_t\big\|_{\infty} \leq \delta \bigg\} \bigg. \nn\\*
&\qquad\qquad\qquad\qquad\bigg. +\mathrm{Pr}\bigg\{ \omega_{1,t}\in\hat{\Omega}_1~\mathrm{and}~\exists~(\hatbs,\hatbv)\in\big(\hatbs,\hatbv\big)^k:\big\|\hatbv-\bV_t\big\|_{\infty} \leq \frac{2\delta}{n} \bigg\}\bigg) \label{Omega-6}\\
&\leq \varepsilon+\sum_{t\in[k]}\bigg(\mathrm{Pr}\bigg\{ \exists~(\hatbs,\hatbv)\in\big(\hatbs,\hatbv\big)^k:\omega_{0,t} \neq \mathrm{q_s}(\hatbs)~\mathrm{and}~\big\|\hatbs-\bS_t\big\|_{\infty} \leq \delta \bigg\} \bigg. \nn\\*
&\qquad\qquad\qquad\qquad\bigg. +\mathrm{Pr}\bigg\{\exists~(\hatbs,\hatbv)\in\big(\hatbs,\hatbv\big)^k:\omega_{1,t} \neq \mathrm{q_v}(\hatbv)~\mathrm{and}~\big\|\hatbv-\bV_t\big\|_{\infty} \leq \frac{2\delta}{n} \bigg\}\bigg) \label{Omega-7}\\
&\leq \varepsilon+\sum_{t\in[k]}\sum_{(\hatbs,\hatbv)\in(\hatbs,\hatbv)^k}\bigg(\mathrm{Pr}\bigg\{ \omega_{0,t} \neq \mathrm{q_s}(\hatbs)~\mathrm{and}~\big\|\hatbs-\bS_t\big\|_{\infty} \leq \delta \bigg\}+\mathrm{Pr}\bigg\{\omega_{1,t} \neq \mathrm{q_v}(\hatbv)~\mathrm{and}~\big\|\hatbv-\bV_t\big\|_{\infty} \leq \frac{2\delta}{n} \bigg\}\bigg) \label{Omega-8}\\
&\leq \varepsilon+\sum_{t\in[k]}\sum_{(\hatbs,\hatbv)\in(\hatbs,\hatbv)^k}\bigg(2d\delta\tilde{M}+\frac{4d\delta}{n}\tilde{M}_*\bigg) \label{Omega-9}\\
&\leq \varepsilon+2k^2d\delta\bigg(\tilde{M}+\frac{2}{n}\tilde{M}_*\bigg) \label{Omega-10}\\
&\leq \varepsilon+2(1+4v_+)k^2d\beta , \label{Omega-rlt}
\end{align}
where \eqref{Omega-3} follows from \eqref{pe-epsilon}, \eqref{Omega-9} follows from the union bound and arguments analogous to those in~\cite{kaspi2017searching} and \eqref{Omega-rlt} follows from the definition of $\tilde{M}$ and $\tilde{M}_*$.

For subsequent analysis, let $\Delta:[\tilde{M}^d]^k \times [\tilde{M}_*^d]^k \rightarrow [(\tilde{M} \times \tilde{M}_*)^{d \times k}]$ be an arbitrary one-to-one mapping from $2~(d \times k)$-dimensional vectors to an integer. To connect the current problem to data transmission, let $\Omega=\Delta(\Omega_{[0:1]})$ be a random variable. Similarly, we can define $\hat{\Omega}=\Delta(\hat{\Omega}_{[0:1]})$. It follows from \eqref{Omega-rlt} that
\begin{align}
\mathrm{Pr}\big\{\hat{\Omega} \neq \Omega\big\}=\mathrm{Pr}\big\{\Omega_{[0:1]} \neq \hat{\Omega}_{[0:1]}\big\}.
\end{align}
Since we consider uniformly distributed initial location and velocity pairs $(\bS,\bV)^k$, the random variable $\Omega$ is uniformly distributed over $[(\tilde{M} \times \tilde{M}_*)^{d \times k}]$. Given any queries $\calA^n$, the probability $\mathrm{Pr}\{\hat{\Omega} \neq \Omega\}$ is the average error probability for a channel coding problem with states at both the encoder and decoder, where at each time $i\in[n]$, the channel is given by the query-dependent channel $P_{Y|Z}^{h(|\calA_i|)}$.

\subsubsection{Final Steps}

Using the non-asymptotic converse bound for channel coding, for any $\kappa\in(0,1-\varepsilon-2(1+4v_+)k^2d\beta)$, we have
\begin{align}
&\log\big(\tilde{M} \times \tilde{M}_*\big)^{d \times k} \nn\\
&\leq \inf_{Q_{Y^n}\in\calP(\calY^n)}\sup_{z\in\calZ}\sup\bigg\{r\in\bbR_+ | \mathrm{Pr}\bigg\{\sum_{i\in[n]}\log\frac{P_{Y|Z}^{h(|\calA_i|)}(Y_i|Z_i)}{Q_{Y}(Y_i)} \leq r\bigg\} \leq \varepsilon+2(1+4v_+)k^2d\beta+\kappa\bigg\}-\log\kappa \\
&\leq \sup_{z\in\calZ}\sup\bigg\{r\in\bbR_+ | \mathrm{Pr}\bigg\{\sum_{i\in[n]}\log\frac{P_{Y|Z}^{h(|\calA_i|)}(Y_i|Z_i)}{P_{Y}^{h(|\calA_i|)}(Y_i)} \leq r\bigg\} \leq \varepsilon+2(1+4v_+)k^2d\beta+\kappa\bigg\}-\log\kappa , \label{logtildeM-1}
\end{align}
where \eqref{logtildeM-1} follows by choosing $Q_{Y^n}$ as distribution $\prod_{i\in[n]} P_{Y}^{f(|\calA_i|)}$.

Recall the definitions of $\tilde{M}$ and $\tilde{M}_*$. It follows that
\begin{align}
\log\big(\tilde{M} \times \tilde{M}_*\big)^{d \times k} &= dk\big(\log\tilde{M}+\log\tilde{M}_*\big) \\
&=dk\big(2\log\tilde{M}+\log(2nv_+)\big) \\
&=2dk\log\beta-2dk\log\delta+dk\log(2nv_+). \label{logtildeM-2}
\end{align}

Combining \eqref{logtildeM-1} and \eqref{logtildeM-2} leads to
\begin{align}
&-\log\delta \nn\\
&\leq \sup_{z\in\calZ}\frac{1}{2dk}\bigg\{\sup\bigg\{r\in\bbR_+ | \mathrm{Pr}\bigg\{\sum_{i\in[n]}\imath_{\calA_i}(Z_i;Y_i) \leq r\bigg\} \leq \varepsilon+2(1+4v_+)k^2d\beta+\kappa\bigg\}-dk\log\big(2nv_+\beta^2\big)-\log\kappa\bigg\} \\
&= \sup_{\calA^n \in ([0,1]^d)^n}\frac{1}{2dk}\bigg\{\sup\bigg\{r\in\bbR_+ | \mathrm{Pr}\bigg\{\sum_{i\in[n]}\imath_{\calA_i}(Z_i;Y_i) \leq r\bigg\} \leq \varepsilon+2(1+4v_+)k^2d\beta+\kappa\bigg\}-dk\log\big(2nv_+\beta^2\big)-\log\kappa\bigg\}.
\end{align}

The proof of non-asymptotic converse bound (Theorem \ref{th3}) is now completed.

\section{Proof of Second-Order Asymptotics (Theorem \ref{th4})} 
\label{sec_proof_a}

\subsection{Achievability} \label{Achievable Second-Order Asymptotic Bound}

From Theorem \ref{th1}, we have that there exists a non-adaptive query procedure such that the average excess-resolution probability is upper bounded as follows:
\begin{align}
&\mathbb{E}\big[\mathrm{P_e}\big((\bS,\bV)^k,\mathbf{X}\big)\big] \nn\\
&\leq 4n\exp\big(-2n^dM^d\eta^2\big) + \exp\big(n \eta Lc(h(p))\big) \times \bigg( \underset{P_{\mathbf{X}Y^n}^{\mathrm{alt}}}{\mathrm{Pr}}\big\{ \calG_1(n,\gamma) \big\} + \bbH_1\Big(\big((2nv_++3)n^4M^2\big)^d-k,k,\gamma\Big) \bigg. \nn\\
&\qquad \bigg.+ \sum_{\calJ\in\Lambda([k])} \bigg( \underset{P_{\mathbf{X}Y^n}^{\mathrm{alt}}}{\mathrm{Pr}}\big\{ \calG_2\big(n,\calJ,\lambda_{\calJ}\big) \big\} + \bbH_1\Big(\big((2nv_++3)n^4M^2\big)^d-k,k-|\calJ|,\gamma-nC(p,|\calJ|)-n\lambda_\calJ\Big) \bigg) \bigg). \label{epe-con-2}
\end{align}

In subsequent analysis, the probability terms are calculated with respect to $P_{\mathbf{X}Y^n}^{\mathrm{alt}}$ (cf. \eqref{P_alt}) unless otherwise stated. Recall the definitions of $C(p,k)$ in \eqref{cempty} and $V(p,k)$ in \eqref{vempty}. We choose $M\in\bbN_+$ and $\gamma\in\mathbb{R}_+$ such that for some $\varepsilon^{\prime}\in(0,1)$,
\begin{align}
\log M&=\frac{nC(p,k)+\sqrt{nV(p,k)}\Phi^{-1}(\varepsilon^{\prime})-4dk\log n-nv_+}{2dk} , \label{logm}\\
\gamma&=dk\log n^4M^2 + nv_+ = nC(p,k)+\sqrt{nV(p,k)}\Phi^{-1}(\varepsilon^{\prime}). \label{gamma}
\end{align}

Given $(p,t)\in(0,1)\times[k]$, recall the definition of $C(p,t)$ in \eqref{cempty}. Let $T(p,t)$ denote the third absolute moment of the mutual information density $\imath^{h(p)}$ in \eqref{mu_info_e}, i.e.
\begin{align}
T(p,t):=\mathbb{E}\big[\big|\imath^{h(p)}\big(X_{[t]};Y\big)-C(p,t)\big|^3\big]. \label{moment}
\end{align}
Recall the definition of event $\calG_1(\cdot)$ in \eqref{calG1}. The Berry-Esseen theorem~\cite{berry1941accuracy,esseen1942liapounoff} implies that
\begin{align}
\mathrm{Pr}\big\{ \calG_1(n,\gamma) \big\} &= \mathrm{Pr}\Big\{ \imath^{h(p)}\big(X^n(1),\ldots,X^n(k);Y^n\big) < \gamma \Big\} \\
&= \mathrm{Pr}\Big\{ \imath^{h(p)}\big(X^n(1),\ldots,X^n(k);Y^n\big) < dk\log n^4M^2 + nv_+ \Big\} \\
&= \mathrm{Pr}\Big\{ \imath^{h(p)}\big(X^n(1),\ldots,X^n(k);Y^n\big) < nC(p,k)+\sqrt{nV(p,k)}\Phi^{-1}(\varepsilon^{\prime}) \Big\} \\
&\leq \varepsilon^{\prime} + \frac{6T(p,k)}{\big(nV(p,k)\big)^{\frac{3}{2}}}. \label{epsilon'+t}
\end{align}

The weak law of large numbers implies that when $n \rightarrow \infty$, for any $\calJ\in\Lambda([k])$,
\begin{align}
\mathrm{Pr}\big\{\calG_2\big(n,\calJ,\lambda_{\calJ}\big)\big\} = \mathrm{Pr}\Big\{\frac{1}{n}\imath^{h(p)}\big(X_{\calJ}^n;Y^n\big) - C(p,|\calJ|) > \lambda_\calJ\Big\}=\theta_1(n) \rightarrow 0. \label{khinchin}
\end{align}

Then, we analyze the last part of the average excess-resolution probability in \eqref{epe-con-2}. For any $\calJ\in\Lambda([k])$,
\begin{align}
&\bbH_1\Big(\big((2nv_++3)n^4M^2\big)^d-k,k-|\calJ|,\gamma-nC(p,|\calJ|)-n\lambda_\calJ\Big) \nn\\
&= \binom{\big((2nv_++3)n^4M^2\big)^d-k}{k-|\calJ|}\exp\Big(n\big(C(p,|\calJ|)-C(p,k)+\lambda_\calJ\big)-\sqrt{nV(p,k)}\Phi^{-1}(\varepsilon^{\prime})\Big) = \theta_2(n). \label{lambda_calJ}
\end{align}
Consider that
\begin{align}
C(p,k)-C(p,|\calJ|) &= \mathbb{E}\big[\imath^{h(p)}\big(X_{[k]};Y\big)\big] - \mathbb{E}\big[\imath^{h(p)}\big(X_{\calJ};Y\big)\big] \\
&= \mathbb{E}\bigg[\log \frac{P_{Y | X_{[k]}}^{h(p)}\big(y | x_{[k]}\big)}{P_Y^{h(p)}(y)} - \log \frac{P_{Y | X_{\calJ}}^{h(p)}\big(y | x_{\calJ}\big)}{P_Y^{h(p)}(y)}\bigg] \\
&= \mathbb{E}\bigg[\log \frac{P_{Y | X_{[k]}}^{h(p)}\big(y | x_{[k]}\big)}{P_{Y | X_{\calJ}}^{h(p)}\big(y | x_{\calJ}\big)}\bigg] \\
&= \mathbb{E}\big[\imath_{\calJ}^{h(p)}\big(X_{[k]};Y\big)\big].
\end{align}
When choosing some value of the auxiliary variable $\lambda_\calJ$ to satisfy $0 < \lambda_\calJ \leq \mathbb{E}[\imath_{\calJ}^{h(p)}(X_{[k]};Y)]$, we have that $\theta_2(n) \rightarrow 0$ as $n \rightarrow \infty$.

Combining \eqref{epe-con-2}, \eqref{epsilon'+t}, \eqref{khinchin} and \eqref{lambda_calJ}, when $M$ and $\gamma$ satisfy \eqref{logm} and \eqref{gamma}, respectively, we have
\begin{align}
&\mathbb{E}\Big[\mathrm{P_e}\Big(\big(\bS,\bV\big)^k,\mathbf{X}\Big)\Big] \nn\\*
&\leq 4n\exp\big(-2n^dM^d\eta^2\big) + \exp\big(n \eta Lc(h(p))\big) \times \Bigg(\varepsilon^{\prime} + \frac{6T(p,k)}{(nV(p,k))^{\frac{3}{2}}} + \frac{\binom{((2nv_++3)n^4M^2)^d-k}{k}}{(n^4M^2)^{dk}\exp(nv_+)} + \theta_1(n) + \theta_2(n) \Bigg). \label{epevarepsilonprime}
\end{align}
The result in \eqref{epevarepsilonprime} implies that there exist binary vectors $\bar{\mathbf{x}}=(\bar{x}^n(1),\ldots,\bar{x}^n((nM)^d))$ such that the excess-resolution probability with respect to the resolution level $\frac{2}{M}$ for any initial location and velocity pairs $(\bS,\bV)^k=\{(\bS_1,\bV_1),\ldots,(\bS_k,\bV_k)\}$ is upper bounded by the right-hand side of \eqref{epevarepsilonprime}.

Let
\begin{align}
\eta=\sqrt{\frac{d\log M}{2M^d}} , \label{eta}
\end{align}
and recall the choice of $M$ in \eqref{logm}. We can verify that as $n \rightarrow \infty$,
\begin{align}
\exp\big(n \eta Lc(h(p))\big) &= 1 + n \eta Lc(h(p)) + o\big(n \eta Lc(h(p))\big) \label{taylor}\\*
&= 1 + O\Bigg(\frac{n^{\frac{3}{2}}}{\exp\Big(\frac{n(C(p,k)-v_+)}{4k}\Big)}\Bigg) \rightarrow 1 ,
\end{align}
and
\begin{align}
4n\exp\big(-2n^dM^d\eta^2\big) &= 4n\exp\big(-n^d\log M^d\big) \\*
&= \frac{4n}{(M^d)^{n^d}} \\*
&= O\Bigg(\frac{n}{\exp\Big(\frac{n^{d+1}(C(p,k)-v_+)}{2k}\Big)}\Bigg) \rightarrow 0 ,
\end{align}
where \eqref{taylor} follows from the Taylor expansion.

Furthermore, in the query-dependent AWGN channel case, 
\begin{align}
\exp\big(\Xi(n,p,\sigma,\eta)\big) &= 1 + \Xi(n,p,\sigma,\eta) + o\big(\Xi(n,p,\sigma,\eta)\big) \label{taylor-2}\\*
&= 1 + O\bigg(n \eta L\bigg(\frac{1}{h(p)-\eta L}+\frac{(2h(p)+\eta L)(h(p)^2+\eta L(2h(p)+\eta L))}{h(p)^2(h(p)^2-\eta L(2h(p)+\eta L))}\bigg)\bigg) \rightarrow 1 ,
\end{align}
and
\begin{align}
\exp\Big(-\frac{n(1-\log 2)}{2}\Big) = \Big(\frac{e}{2}\Big)^{-\frac{n}{2}} \rightarrow 0.
\end{align}

Since $\calX$ and $\calY$ are both finite sets, we conclude that the dispersion $V(p,t)$ and the third absolute moment $T(p,t)$ are both finite for any $t\in[k]$. We choose $\varepsilon^{\prime}$ such that the bound in \eqref{epevarepsilonprime} is exactly $\varepsilon\in(0,1)$. The above analysis implies that
\begin{align}
\varepsilon^{\prime}=\varepsilon-\Theta\Big(\frac{1}{n^\frac{3}{2}}\Big).
\end{align}

Note that the above result holds for any $p\in(0,1)$. Thus, for any $n\in\bbN$ and $\varepsilon\in(0,1)$, the minimal achievable resolution $\delta(n,k,d,\varepsilon)$ satisfies
\begin{align}
-\log \delta(n,k,d,\varepsilon) &\geq \max_{p\in(0,1)} \frac{nC(p,k)+\sqrt{nV(p,k)}\Phi^{-1}(\varepsilon^{\prime})-4dk\log n-nv_+}{2dk} - \log2 \label{secorder-rlt1}\\
&= \max_{p\in(0,1)} \frac{1}{2dk} \Big(nC(p,k)+\sqrt{nV(p,k)}\Phi^{-1}(\varepsilon) - 4dk\log n-nv_+ + O(1)\Big) - \log2 \label{secorder-rlt2}\\
&= \frac{nC(k)+\sqrt{nV(k,\varepsilon)}\Phi^{-1}(\varepsilon)-4dk\log n-nv_+ + O(1)}{2dk} , \label{secorder-rlt3}
\end{align}
where \eqref{secorder-rlt2} follows from the Taylor expansion of $\Phi^{-1}(\cdot)$ and \eqref{secorder-rlt3} follows from the definition of $C(k)$ in \eqref{ck} and $V(k,\varepsilon)$ in \eqref{vk}, respectively.

The achievability proof of Theorem \ref{th4} is now completed.

\subsection{Converse} \label{Second-Order Asymptotic Converse Bound}

For any queries $\calA^n\in([0,1]^d)^n$, define the following first-, second- and third-order moments of the information density $\imath_{\calA_i}(Z_i;Y_i)$:
\begin{align}
C_{\calA^n}&:=\frac{1}{n}\sum_{i\in[n]} \mathbb{E}\big[\imath_{\calA_i}(Z_i;Y_i)\big] , \\
V_{\calA^n}&:=\frac{1}{n}\sum_{i\in[n]} \mathrm{Var}\big[\imath_{\calA_i}(Z_i;Y_i)\big] , \\
T_{\calA^n}&:=\frac{1}{n}\sum_{i\in[n]} \mathbb{E}\Big[\big|\imath_{\calA_i}(Z_i;Y_i)-\mathbb{E}\big[\imath_{\calA_i}(Z_i;Y_i)\big]\big|^3\Big].
\end{align}
Since we consider finite input and output alphabets, the moment $C_{\calA^n},V_{\calA^n}$ and $T_{\calA^n}$ are all finite.

Consider those $\calA^n$ such that $V_{\calA^n}$ is strictly positive, i.e., there exists $V_{\_} > 0$ satisfying $V_{\_} \leq V_{\calA^n}$. For any $r\in\bbR_+$, using the Berry-Esseen theorem, we have that
\begin{align}
\mathrm{Pr}\bigg\{\sum_{i\in[n]}\imath_{\calA_i}(Z_i;Y_i) \leq r\bigg\} &\geq \Phi\Big(\frac{r-nC_{\calA^n}}{\sqrt{nV_{\calA^n}}}\Big)-\frac{6T_{\calA^n}}{\sqrt{n}\big(\sqrt{V_{\calA^n}}\big)^3} \\*
&\geq \Phi\Big(\frac{r-nC_{\calA^n}}{\sqrt{nV_{\calA^n}}}\Big)-\frac{6T_{\calA^n}}{\sqrt{n}\big(\sqrt{V_{\_}}\big)^3}.
\end{align}
Thus, for any $\varepsilon\in(0,1)$,
\begin{align}
\sup\bigg\{r\in\bbR_+:\mathrm{Pr}\bigg\{\sum_{i\in[n]}\imath_{\calA_i}(Z_i;Y_i) \leq r\bigg\} \leq \varepsilon\bigg\} &\leq \sup\bigg\{r\in\bbR_+:\Phi\Big(\frac{r-nC_{\calA^n}}{\sqrt{nV_{\calA^n}}}\Big)-\frac{6T_{\calA^n}}{\sqrt{n}\big(\sqrt{V_{\_}}\big)^3} \leq \varepsilon\bigg\} \\
&\leq nC_{\calA^n}+\sqrt{nV_{\calA^n}}\Phi^{-1}\bigg(\varepsilon+\frac{6T_{\calA^n}}{\sqrt{n}\big(\sqrt{V_{\_}}\big)^3}\bigg).
\end{align}

Therefore, choosing $\beta=\frac{1}{\sqrt{n}}$ and $\kappa=\frac{1}{\sqrt{n}}$, we have
\begin{align}
\nn&\sup_{\calA^n}\sup\bigg\{r\in\bbR_+|\mathrm{Pr}\bigg\{\sum_{i\in[n]}\imath_{\calA_i}(Z_i;Y_i) \leq r\bigg\} \leq \varepsilon+2(1+4v_+)k^2d\beta+\kappa\bigg\} \\*
&\leq \sup_{\calA^n}\bigg\{ nC_{\calA^n}+\sqrt{nV_{\calA^n}}\Phi^{-1}\bigg(\varepsilon+2(1+4v_+)k^2d\beta + \kappa+\frac{6T_{\calA^n}}{\sqrt{n}(\sqrt{V_{\_}})^3}\bigg) \bigg\} \\*
&= \sup_{\calA^n:|\calA_i|\in\calP_{k},i\in[n]}\bigg\{ nC_{\calA^n}+\sqrt{nV_{\calA^n}}\Phi^{-1} \times \bigg(\varepsilon+2(1+4v_+)k^2d\beta+\kappa+\frac{6T_{\calA^n}}{\sqrt{n}(\sqrt{V_{\_}})^3}\bigg) \bigg\} + O(1) , \label{supsup-rlt}
\end{align}
where \eqref{supsup-rlt} follows from~\cite[Lemma 49]{polyanskiy2010channel}.

Recall the definitions of $C(p,t)$ in \eqref{cempty}, $V(p,t)$ in \eqref{vempty} and $C(k)$ in \eqref{ck}, respectively. Recall the capacity achieving optimizer $p^*$. It follows that
\begin{align}
\sup_{\calA^n\in([0,1]^d)^n} C_{\calA^n}&=\frac{1}{n}\sum_{i\in[n]} \mathbb{E}\big[\imath_{\calA_i}(Z_i;Y_i)\big] \\
&= \frac{1}{n}\sum_{i\in[n]} \mathbb{E}\Big[\imath^{|\calA_i|}\big(X_{[k]};Y\big)\Big] \\
&= \frac{1}{n}\sum_{i\in[n]} C(|\calA_i|,k) \\
&\leq \sup_{\calA\in[0,1]^d} C(|\calA_i|,k) \\
&= \max_{p\in(0,1)} C(p,k) \\
&= C(k) , \label{sup-ca}
\end{align}
and
\begin{align}
\sup_{\calA^n\in([0,1]^d)^n} V_{\calA^n}&=\frac{1}{n}\sum_{i\in[n]} \mathrm{Var}\big[\imath_{\calA_i}(Z_i;Y_i)\big] \\
&= \frac{1}{n}\sum_{i\in[n]} \mathrm{Var}\Big[\imath^{|\calA_i|}(X_{[k]};Y)\Big] \\
&= \frac{1}{n}\sum_{i\in[n]} V(|\calA_i|,k) \\
&= V(p^*,k). \label{sup-va}
\end{align}

Therefore, we have
\begin{align}
&-2dk\log\delta \nn\\*
&\leq \sup_{\calA^n}\bigg\{\sup\bigg\{r\in\bbR_+|\mathrm{Pr}\bigg\{\sum_{i\in[n]}\imath_{\calA_i}(Z_i;Y_i) \leq r\bigg\} \leq \varepsilon+2(1+4v_+)k^2d\beta+\kappa\bigg\}-dk\log\big(2nv_+\beta^2\big)-\log\kappa\bigg\} \label{2nd-converse-1}\\*
&\leq \sup_{\calA^n:|\calA_i|\in\calP_{k},i\in[n]}\bigg\{ nC_{\calA^n}+\sqrt{nV_{\calA^n}}\Phi^{-1}\bigg(\varepsilon+2(1+4v_+)k^2d\beta+\kappa+\frac{6T_{\calA^n}}{\sqrt{n}(\sqrt{V_{\_}})^3}\bigg)-dk\log\big(2nv_+\beta^2\big)-\log\kappa \bigg\} + O(1) \label{2nd-converse-2}\\
&\leq \sup_{\calA^n:|\calA_i|\in\calP_{k},i\in[n]}\bigg\{ nC_{\calA^n}+\sqrt{nV_{\calA^n}}\Phi^{-1}\bigg(\varepsilon+\frac{2(1+4v_+)k^2d+1}{\sqrt{n}}+\frac{6T_{\calA^n}}{\sqrt{n}(\sqrt{V_{\_}})^3}\bigg)+\frac{1}{2}\log n - dk\log(2v_+) \bigg\} + O(1) \label{2nd-converse-3}\\
&\leq \sup_{\calA^n:|\calA_i|\in\calP_{k},i\in[n]}\bigg\{ nC_{\calA^n}+\sqrt{nV_{\calA^n}}\Phi^{-1}(\varepsilon)+O(1)+\frac{1}{2}\log n - dk\log(2v_+) \bigg\} + O(1) \label{2nd-converse-4}\\
&\leq nC(k)+\sqrt{nV(k,\varepsilon)}\Phi^{-1}(\varepsilon)+\frac{1}{2}\log n - dk\log(2v_+)+O(1) \label{2nd-converse-5}\\
&= nC(k)+\sqrt{nV(k,\varepsilon)}\Phi^{-1}(\varepsilon)+O(\log n) , \label{2nd-converse-rlt}
\end{align}
where \eqref{2nd-converse-2} follows from \eqref{supsup-rlt}, \eqref{2nd-converse-4} follows from the Taylor expansion for $\Phi^{-1}(\cdot)$ at around $\varepsilon$ and \eqref{2nd-converse-5} follows from the results in \eqref{sup-ca} and \eqref{sup-va}.

The converse proof of Theorem \ref{th4} is now completed.

\section{Conclusion}

We have proposed a non-adaptive query procedure for tracking multiple moving targets and showed that this procedure theoretically achieves second-order asymptotic bounds on the achievable resolution. Our proposed query procedure uses the single threshold rule on the mutual information density, which reduces the computational complexity without compromise in performance, as compared with the algorithm proposed for multiple stationary targets~\cite{zhou2023resolution}. Furthermore, we discussed two special cases with either the initial location or the moving velocity is known. We also generalized our results to account for a piecewise constant velocity model~\cite{zhou2023resolution}. Finally, we demonstrated how the theoretical results of twenty questions estimation can be used for the practical problem of beam tracking and provided a numerical example for the application of beam tracking in 5G wireless communication. Numerical experiments in this paper have shown that when the number of queries is too small, the $O(\log n)$ approximation given in \eqref{th4-rlt} and \eqref{th4-rlt-2} of Theorem \ref{th4} is no longer valid as a lower bound, as the proposed beam tracking algorithm achieves resolution that violate the second-order asymptotic bound (blue curve in Fig. \ref{fig56}). This suggests that refining the analysis of this paper to include a third-order term in \eqref{th4-rlt} and \eqref{th4-rlt-2} would be worthwhile in the regime of a small number of queries.

There are several other avenues for future directions. Firstly, we used random coding and information spectrum decoding in the non-adaptive query procedure, which has very high complexity. In future, one can propose practical query procedures~\cite{chen2005reduced,afisiadis2014low} that achieve our derived theoretical benchmarks. Secondly, compared to the non-adaptive procedure, the performance advantage of the adaptive algorithm is superior~\cite{sun2023achievable}. It would be of interest to generalize our analyses to adaptive query procedures~\cite{chiu2019noisy} and explore the benefit of adaptivity. It is also worthwhile to design and analyze adaptive beam tracking algorithms for multiple moving targets by draw inspiration from~\cite{ronquillo2023integrated}. Finally, one can  explore other applications of twenty questions estimation, such as reconfigurable intelligent surface defective detection~\cite{zhang2024target} and delayed Doppler domain channel estimation for integrated sensing and communication~\cite{yuan2024otfs}.

\appendix 

\subsection{Proof Sketch of Prior Information on Target Velocities (Theorem \ref{caseA-2order})} 
\label{Specialization Bound 1}

\begin{algorithm}[ht]
\caption{Non-Adaptive Query Procedure for the Initial Location Search of Multiple Moving Targets with Known Velocities}
\label{alg:2} 
\begin{algorithmic}[0]  
\REQUIRE
The number of queries $n\in\bbN$, dimension $d\in\bbN$, the vector of targets' velocities $(\bv_1,\ldots,\bv_k)\in([-v_+,v_+]^d)^k$ and three design parameters $(M,p,\gamma)\in\bbN \times (0,1) \times \bbR_+$
\ENSURE
A set of estimated targets' initial locations $\hat{\calS}^k\in([0,1]^d)^k$
~\\
\hrule 
~\\
\STATE Divide the $d$-dimensional unit cube $[0,1]^d$ into $M^d$ equally sized non-overlapping sub-cubes $(\calC_1,\ldots,\calC_{M^d})$.
\STATE \textbf{\emph{Query generation and response collection:}}
\STATE Generate $n$ binary vectors $(x^n(1),\ldots,x^n(M^d))$ with length-$M^d$, where each vector is generated i.i.d. from Bern$(p)$.
\FOR{$i\in[n]$}
\STATE Form a query region $\calA_i$ as
\begin{align}
\calA_i:=\bigcup_{j\in[M^d]: x_{i,j}=1} \calC_j \nn.
\end{align}
\STATE Pose the $i$-th query about whether there are targets in the region $\calA_i$ and obtain a noisy response $y_i$ from the oracle.
\ENDFOR
\STATE Collect noisy responses $y^n=(y_1,\ldots,y_n)$.
\STATE \textbf{\emph{Decoding:}}
\IF{there exists a set of trajectories $\{\hat{\bw}_1^n,\ldots,\hat{\bw}_k^n\}\in\calB_{n,M}(\bv_1) \times \ldots \times \calB_{n,M}(\bv_k)$ such that $$\imath^{h(p)}\big(x^n(\Gamma(\hat{\bw}_1^n)),\ldots,x^n(\Gamma(\hat{\bw}_k^n)); y^n\big) \geq \gamma,$$}
\RETURN estimates $\{\hatbs_1,\ldots,\hatbs_k\}$ such that for any $t\in[k]$, $$w(\hatbs_t,\bv_t,[n])=\hat{\bw}_t^n.$$
\ELSE
\RETURN estimates $\{ \hatbs_1,\ldots,\hatbs_k \}$ such that $\hatbs_t = \mathbf{0.5}_{d \times 1}$ for each $t\in[k]$.
\ENDIF
\end{algorithmic}
\end{algorithm}

We still consider the same three excess-resolution events $\calE_1,\calE_2$ and $\calE_3$ mentioned in Section \ref{Achievable Non-Asymptotic Bound DMC}. Recall the definitions of events $\calG_1(\cdot),\calG_2(\cdot)$ in \eqref{calG1}, \eqref{calG2} and functions $\bbH_1(\cdot),\bbH_2(\cdot)$ in \eqref{calH1}, \eqref{calH2}. Given any $k$ targets' initial locations $\bs^k\in([0, 1]^d)^k$, for any query matrix $\mathbf{X}\in(\{0,1\}^n)^{M^d}$ and noisy responses $Y^n\in\calY^n$, when targets’ velocities $\bv^k\in([-v_+, v_+]^d)^k$ are known, the respective probabilities are as follows:
\begin{align}
\underset{P_{\mathbf{X}Y^n}^{\mathrm{alt}}}{\mathrm{Pr}}\big\{\calE_1\big(\bs^k,\bv^k,\mathbf{X},Y^n\big)\big\} &= \underset{P_{\mathbf{X}Y^n}^{\mathrm{alt}}}{\mathrm{Pr}}\big\{ \calG_1(n,\gamma) \big\} , \\
\underset{P_{\mathbf{X}Y^n}^{\mathrm{alt}}}{\mathrm{Pr}}\big\{\calE_2\big(\bs^k,\bv^k,\mathbf{X},Y^n\big)\big\} &\leq \bbH_2\big(M^d-1,k,\gamma\big) , \\
\underset{P_{\mathbf{X}Y^n}^{\mathrm{alt}}}{\mathrm{Pr}}\big\{\calE_3\big(\bs^k,\bv^k,\mathbf{X},Y^n\big)\big\} &\leq
\sum_{\calJ\in\Lambda([k])} \bigg(\underset{P_{\mathbf{X}Y^n}^{\mathrm{alt}}}{\mathrm{Pr}}\big\{\calG_2(n,\calJ,\lambda_{\calJ})\big\} + \bbH_2\big(M^d-1,k-|\calJ|,\gamma-nC(p,|\calJ|)-n\lambda_\calJ\big)\bigg).
\end{align}

Therefore, given $(n, k, d)\in\bbN^3$, fixed constants $(\gamma,\lambda_{\calJ})\in\bbR_+^2$ for $\calJ\in\Lambda([k])$, for any $(M, p, \eta)\in\bbN \times (0,1) \times \bbR_+$, we have the achievable non-asymptotic bound for an $\big(n,k,d,\frac{1}{M},\varepsilon\big)$-non-adaptive query procedure with known velocities under any query-dependent DMC satisfying \eqref{dmc} such that
\begin{align}
\varepsilon &\leq 4n\exp\big(-2M^d\eta^2\big) + \exp\big(n \eta Lc(h(p))\big) \times \bigg( \underset{P_{\mathbf{X}Y^n}^{\mathrm{alt}}}{\mathrm{Pr}}\big\{ \calG_1(n,\gamma) \big\} + \bbH_2\big(M^d-1,k,\gamma\big) \bigg. \nn\\ 
&\qquad\qquad\qquad\qquad \bigg.+ \sum_{\calJ\in\Lambda([k])} \bigg( \underset{P_{\mathbf{X}Y^n}^{\mathrm{alt}}}{\mathrm{Pr}}\big\{\calG_2\big(n,\calJ,\lambda_{\calJ}\big)\big\}+\bbH_2\big(M^d-1,k-|\calJ|,\gamma-nC(p,|\calJ|)-n\lambda_\calJ\big) \bigg) \bigg) ,
\end{align}
which shows that the excess-resolution probability is not affected by the maximum speed when targets' velocities are known. In other words, the multiple moving target case is equivalent to the multiple stationary targets case when targets' velocities are known because of the strong correspondence between velocity and trajectory.

Then, recall the definitions of $C(p,k)$ in \eqref{cempty} and $V(p,k)$ in \eqref{vempty}. We choose $M$ and $\gamma$ such that for some $\varepsilon^{\prime}\in(0,1)$,
\begin{align}
\log M &= \frac{1}{dk}\Big(nC(p,k)+\sqrt{nV(p,k)}\Phi^{-1}(\varepsilon^{\prime})-\frac{1}{2}\log n\Big) , \\
\gamma &= dk\log M + \frac{1}{2}\log n ,
\end{align}
and combine the same idea of data transmission over a point-to-point channel used in~\cite[Section III-C]{zhou2022resolution}. For any $(n, k, d, \varepsilon)\in\bbN^3 \times (0,1)$, we can obtain the second-order asymptotic bound as follows:
\begin{align}
-\log \delta^*(n,k,d,\varepsilon) = \frac{nC(k)+\sqrt{nV(k,\varepsilon)}\Phi^{-1}(\varepsilon)+\Theta(\log n)}{dk}.
\end{align}

\subsection{Proof Sketch for Prior Information on Target Initial Locations (Theorem \ref{caseB-2order})} 
\label{Specialization Bound 2}

\begin{algorithm}[ht]
\caption{Non-Adaptive Query Procedure for the Velocity Search of Multiple Moving Targets with Known Initial Locations}
\label{alg:3} 
\begin{algorithmic}[0]  
\REQUIRE
The number of queries $n\in\bbN$, dimension $d\in\bbN$, the vector of targets' initial locations $(\bs_1,\ldots,\bs_k)\in([0,1]^d)^k$ and three design parameters $(M,p,\gamma)\in\bbN \times (0,1) \times \bbR_+$
\ENSURE
A set of estimated targets' velocities $\hat{\calV}^k\in([-v_+,v_+]^d)^k$
~\\
\hrule 
~\\
\STATE The procedure is precisely identical to that of Algorithm \ref{alg:1}, except for the decoding phase as follows.
\STATE \textbf{\emph{Decoding:}}
\IF{there exists a set of trajectories $\{\hat{\bw}_1^n,\ldots,\hat{\bw}_k^n\}\in\calB_{n,M}(\bs_1) \times \ldots \times \calB_{n,M}(\bs_k)$ such that $$\imath^{h(p)}\big(x^n(\Gamma(\hat{\bw}_1^n)),\ldots,x^n(\Gamma(\hat{\bw}_k^n)); y^n\big) \geq \gamma,$$}
\RETURN estimates $\{\hatbv_1,\ldots,\hatbv_k\}$ such that for any $t\in[k]$, $$w(\bs_t,\hatbv_t,[n])=\hat{\bw}_t^n.$$
\ELSE
\RETURN estimates $\{ \hatbv_1,\ldots,\hatbv_k \}$ such that $\hatbv_t = \mathbf{0}_{d \times 1}$ for each $t\in[k]$.
\ENDIF
\end{algorithmic}
\end{algorithm}

Theorem \ref{caseB-2order} can also be easily proven based on Theorem \ref{caseA-2order}, except that the number of all possible trajectories has changed from $M^d$ to $((2nv_++3)n^3M)^d$ and letting
\begin{align}
\log M &= \frac{nC(p,k)+\sqrt{nV(p,k)}\Phi^{-1}(\varepsilon^{\prime})-3dk\log n -nv_+}{dk} , \\
\gamma &= dk\log n^3M + nv_+.
\end{align}

\subsection{Proof Sketch for Piecewise Constant Velocity Model (Theorem \ref{caseC-2order})} 
\label{Generalization Bound}

\begin{algorithm}[ht]
\caption{Non-Adaptive Query Procedure for Multiple Moving Target Search under a Piecewise Constant Velocity Model}
\label{alg:4}
\begin{algorithmic}[0]
\REQUIRE
Velocity change time points $\mathbf{n}=(n_1,\ldots,n_B)\in\bbN^B$, dimension $d\in\bbN$, and $B+2$ design parameters $(M,p,\gamma^B)\in\bbN \times (0,1) \times \bbR_+^B$
\ENSURE
A set of estimated trajectories $\{\hat{\bw}_1^n,\ldots,\hat{\bw}_k^n\}$ of $k$ targets with unknown initial locations and piecewise constant velocities\\
~\\
\hrule 
~\\
\STATE Call Algorithm \ref{alg:1} with parameters $(N_1,d,M,p,\gamma_1)$.
\FOR{$j\in[2:B]$}
\STATE Call Algorithm \ref{alg:3} with parameters $(N_j,d,M,p,\gamma_j)$.
\ENDFOR
\end{algorithmic}
\end{algorithm}

Consider the achievable non-asymptotic bounds for the common non-adaptive query procedure proposed in Section \ref{pf Non-Adaptive Query Procedure} and the second special non-adaptive query procedure proposed in Section \ref{sg Known Initial Locations and Unknown Velocities Case}, and recall the definitions of events $\calG_1(\cdot),\calG_2(\cdot)$ in \eqref{calG1}, \eqref{calG2} and functions $\bbH_1(\cdot),\bbH_2(\cdot)$ in \eqref{calH1}, \eqref{calH2}. Given positive integers $(\mathbf{n}, k, d)\in\bbN^{B+2}$, fixed constants $(\gamma^B,\lambda_{\calJ}^B)\in\bbR_+^{2B}$ for $\calJ\in\Lambda([k])$, for any $(M, p, v_+, \eta)\in\bbN \times (0,1) \times (0,0.5) \times \bbR_+$, there exists an $\big(\mathbf{n},k,d,\frac{B+1}{M},\varepsilon\big)$-non-adaptive query procedure under a piecewise constant velocity model and any query-dependent DMC such that
\begin{align}
\varepsilon &\leq 4N_1\exp\big(-2N_1^dM^d\eta^2\big) + \exp\big(N_1 \eta Lc(h(p))\big) \times \bigg( \underset{P_{\mathbf{X}Y^{N_1}}^{\mathrm{alt}}}{\mathrm{Pr}}\big\{ \calG_1\big(N_1,\gamma_1\big) \big\} + \bbH_1\Big(\big((2N_1v_++3)N_1^4M^2\big)^d-k,k,\gamma_1\Big) \bigg. \nn\\*
& \bigg.+ \sum_{\calJ\in\Lambda([k])} \bigg( \underset{P_{\mathbf{X}Y^{N_1}}^{\mathrm{alt}}}{\mathrm{Pr}}\Big\{ \calG_2\big(N_1,\calJ,\lambda_{\calJ,1}\big) \Big\} + \bbH_1\Big(\big((2N_1v_++3)N_1^4M^2\big)^d-k,k-|\calJ|,\gamma_1-N_1C(p,|\calJ|)-N_1\lambda_{\calJ,1}\Big) \bigg) \bigg) \nn\\*
&+\sum_{j\in[2:B]} \bigg( 4N_j\exp\big(-2N_j^dM^d\eta^2\big) + \exp\big(N_j \eta Lc(h(p))\big) \bigg. \nn\\*
&\times \bigg( \underset{P_{\mathbf{X}Y^{N_j}}^{\mathrm{alt}}}{\mathrm{Pr}}\big\{ \calG_1\big(N_j,\gamma_j\big) \big\} + \bbH_2\Big(\big((2N_jv_++3)N_j^3M\big)^d-1,k,\gamma_j\Big) + \sum_{\calJ\in\Lambda([k])} \bigg( \underset{P_{\mathbf{X}Y^{N_j}}^{\mathrm{alt}}}{\mathrm{Pr}}\Big\{ \calG_2\big(N_j,\calJ,\lambda_{\calJ,j}\big) \Big\} \bigg.\bigg. \nn\\*
&\qquad\qquad\qquad\qquad\qquad\qquad\qquad\qquad+ \bigg.\bigg.\bigg. \bbH_2\Big(\big((2N_jv_++3)N_j^3M\big)^d-1,k,\gamma_j-N_jC(p,|\calJ|)-N_j\lambda_{\calJ,j}\Big) \bigg) \bigg) \bigg). \label{common-nonasy-achi}
\end{align}
Furthermore, under the query-dependent AWGN channel for any $\sigma\in\bbR_+$ and each $j\in[B]$, terms $\exp(N_j \eta Lc(h(p)))$ in the bound \eqref{common-nonasy-achi} will be replaced by $\exp(\Xi(N_j,p,\sigma,\eta))$, and an additional term $\sum_{j\in[B]}\exp\big(-\frac{N_j(1-\log 2)}{2}\big)$ will be added.

Recall the non-asymptotic converse bounds for a common non-adaptive query procedure and a second special non-adaptive query procedure. For any $(\mathbf{n}, k, d,\varepsilon)\in\bbN^{B+2} \times (0,1)$, given any $\beta\in\big(0,\frac{1-\varepsilon}{2}\big)$ and any $\kappa\in(0,1-\varepsilon-2(1+4Bv_+)k^2d\beta)$, any $\delta(\mathbf{n},k,d,\varepsilon)$-non-adaptive query procedure under a piecewise constant velocity model satisfies
\begin{align}
&-\log\delta \leq \sup_{\calA^{n_B}\in([0,1]^d)^{n_B}}\frac{1}{(B+1)dk}\bigg\{\sup\bigg\{r\in\bbR_+|\mathrm{Pr}\bigg\{\sum_{i\in[n_B]}\imath_{\calA_i}(Z_i;Y_i) \leq r\bigg\} \leq \varepsilon+2(1+4Bv_+)k^2d\beta+\kappa\bigg\} \bigg. \nn\\
&\qquad\qquad\qquad\qquad\qquad\qquad\qquad\qquad\qquad\bigg. -(B+1)dk\log\beta - dk\log\big(2N_1v_+\big) - dk\sum_{j\in[2:B]}\log\big(2N_jv_+\big)-\log\kappa\bigg\}.
\end{align}

Recall the proof of Theorem \ref{th4} in Section \ref{Achievable Second-Order Asymptotic Bound} and Section \ref{Second-Order Asymptotic Converse Bound}. By choosing $M\in\bbN$ and $\gamma^B\in\bbR_+^B$ for any $j\in[2:B]$ such that
\begin{align}
\log M_1 &= \frac{N_1C(p,k)+\sqrt{N_1V(p,k)}\Phi^{-1}(\varepsilon^{\prime})-4dk\log N_1 - N_1v_+}{2dk} , \\
\log M_j &= \frac{N_jC(p,k)+\sqrt{N_jV(p,k)}\Phi^{-1}(\varepsilon^{\prime})-3dk\log N_j - N_jv_+}{dk} , \\
dk\log M &= \min\Big\{dk\log M_1,\min_{j\in[2:B]}dk\log M_j\Big\} , \\
\gamma_1 &= dk\log N_1^4M_1^2 + N_1v_+ , \\
\gamma_j &= dk\log N_j^3M_j + N_jv_+ .
\end{align}
For any $(\mathbf{n}, k, d, \varepsilon)\in\bbN^{B+2} \times (0,1)$, we can obtain the achievable second-order asymptotic bound under a piecewise constant velocity model as follows:
\begin{align}
-dk\log \delta^*(\mathbf{n},k,d,\varepsilon) &\geq \max_{\substack{(\varepsilon_1,\ldots,\varepsilon_B) \\ \sum_{j\in[B]}\varepsilon_j \leq \varepsilon}} \min\bigg\{ \frac{N_1C(k)+\sqrt{N_1V(k,\varepsilon_1)}\Phi^{-1}(\varepsilon_1)-4dk\log N_1-N_1v_+ + O(1)}{2}, \bigg. \nn\\*
&\quad\bigg. \min_{j\in[2:B]} \bigg\{N_jC(k)+\sqrt{N_jV(k,\varepsilon_j)}\Phi^{-1}(\varepsilon_j)-3dk\log N_j - N_jv_+ + O(1)\bigg\} \bigg\} - dk\log(B+1).
\end{align}

Consider the second-order converse derivation, for any queries $\calA^{n_B}=(\calA_1,\ldots,\calA_{n_B})\in([0,1]^d)^{n_B}$, let
\begin{align}
C_{\calA^{n_B}}&:=\frac{1}{n_B}\sum_{i\in[n_B]} \mathbb{E}\big[\imath_{\calA_i}(Z_i;Y_i)\big] , \\
V_{\calA^{n_B}}&:=\frac{1}{n_B}\sum_{i\in[n_B]} \mathrm{Var}\big[\imath_{\calA_i}(Z_i;Y_i)\big] , \\
T_{\calA^{n_B}}&:=\frac{1}{n_B}\sum_{i\in[n_B]} \mathbb{E}\Big[\big|\imath_{\calA_i}(Z_i;Y_i)-\mathbb{E}[\imath_{\calA_i}(Z_i;Y_i)]\big|^3\Big] ,
\end{align}
and let $V_{\_} \in \bbR_+$ be chosen such that $V_{\_} \leq V_{\calA^{n_B}}$. Using the Berry-Esseen theorem, we have that
\begin{align}
\mathrm{Pr}\bigg\{\sum_{i\in[n_B]}\imath_{\calA_i}(Z_i;Y_i) \leq r\bigg\} \geq \Phi\Big(\frac{r-n_BC_{\calA^{n_B}}}{\sqrt{n_BV_{\calA^{n_B}}}}\Big)-\frac{6T_{\calA^{n_B}}}{\sqrt{n_B}\big(\sqrt{V_{\_}}\big)^3}.
\end{align}
Thus, for any $\varepsilon\in(0,1)$,
\begin{align}
\sup\bigg\{r\in\bbR_+:\mathrm{Pr}\bigg\{\sum_{i\in[n_B]}\imath_{\calA_i}(Z_i;Y_i) \leq r\bigg\} \leq \varepsilon\bigg\} \leq n_BC_{\calA^{n_B}}+\sqrt{n_BV_{\calA^{n_B}}}\Phi^{-1}\bigg(\varepsilon+\frac{6T_{\calA^{n_B}}}{\sqrt{n_B}\big(\sqrt{V_{\_}}\big)^3}\bigg).
\end{align}

Choosing $\beta=\frac{1}{\sqrt{n_B}}$ and $\kappa=\frac{1}{\sqrt{n_B}}$, we can obtain the second-order asymptotic converse bound under a piecewise constant velocity model as follows:
\begin{align}
-(B+1)dk\log \delta^*(\mathbf{n},k,d,\varepsilon) \leq n_BC(k)+\sqrt{n_BV(k,\varepsilon)}\Phi^{-1}(\varepsilon)+O(\log n_B).
\end{align}

\bibliographystyle{IEEEtran}
\bibliography{IEEEabrv,IEEE_cs}
\end{document}